\documentclass[final]{siamltex}

\usepackage{amsmath}
\usepackage{latexsym}
\usepackage{xspace}
\usepackage{amssymb}
\usepackage{graphics}
\usepackage{graphicx}
\usepackage{hyperref}
\usepackage{color}

 {
    \begin{enumerate}}{\end{enumerate}}

\newcommand{\E}{{\bf E}}

\newcommand{\SN}{\smallskip \noindent}

\newcommand{\remove}[1]{}

\providecommand{\expect}[1]{{\sf E}\bigl[ #1 \bigr]}
\providecommand{\prob}[1]{{\sf Pr}\bigl[ #1 \bigr]}

\title{A New Approximation Technique for Resource-Allocation Problems}

\author{Barna Saha\thanks{University of Massachusetts Amherst, Amherst, MA 01002.
         Supported in part by NSF CAREER Award CCF-1652303, NSF Award NSF CCF-0728839, and a Google Faculty Award. ({\tt barna@cs.umass.edu})}
        \and Aravind Srinivasan\thanks{Dept.\ of Computer Science and Institute for Advanced Computer Studies,
University of Maryland, College Park, MD 20742.
Supported in part by NSF CCF-1749864, NSF ITR Award CNS-0426683, NSF Awards CNS-0626636, CNS-1010789 and CCF-1422569, and by research awards from Adobe, Inc.  ({\tt srin@cs.umd.edu})}
}

\begin{document}

\maketitle

\begin{abstract}
We develop a rounding method based on random walks in polytopes, which
leads to improved approximation algorithms and integrality gaps for
several assignment problems that arise in resource allocation
and scheduling. In particular, it generalizes the work of Shmoys \& Tardos
on the generalized assignment problem to the setting where some jobs can be
dropped. New concentration bounds for random bipartite matching are developed as well.   \makeatletter{\renewcommand*{\@makefnmark}{}
\footnotetext{A preliminary version of this paper appeared in the conference proceedings of \emph{Innovations in Computer Science} (ICS) 2010.}\makeatother}
\end{abstract}


\section{Introduction}
\label{sec:intro}
The ``relax-and-round'' paradigm is a well-known approach in
combinatorial optimization. Given an instance of an optimization problem,
we \emph{enlarge} the set of feasible solutions $I$ to some
set $I' \supset I$ -- often a linear-programming (LP) relaxation
of the problem; we then map an (efficiently computed, optimal or
near-optimal)
solution $x^* \in I'$ to some ``nearby'' $x \in I$ and prove that
$x$ is near-optimal in $I$. This second ``rounding'' step is often
a crucial ingredient, and many general techniques have been developed
for it. In this work, we present a new rounding methodology
which leads to several improved approximation algorithms in scheduling,
as well as new concentration-of-measure results.

We start with background on (randomized) rounding and a fundamental
scheduling problem, before describing our contribution.

\subsection{Dependent rounding, iterative rounding, and scheduling}
Recall that in randomized rounding, we use randomization to map
$x^* = (x_1^*, x_2^*, \ldots, x_n^*) \in [0,1]^n$ back to some
$x = (x_1, x_2, \ldots, x_n) \in \{0,1\}^n$ \cite{raghavan-thompson}.
Typically, we choose a value $\alpha$ that is problem-specific,
and -- \emph{independently} for each $i$ -- define $x_i$ to be $1$ with
probability $\alpha x_i^*$, and to be $0$ with the complementary
probability of $1 - \alpha x_i^*$. Independence can, however, lead to
noticeable deviations from the mean for random variables that are
\emph{required} to be very close to (or even be equal to) their mean.
A fruitful idea developed in
\cite{srin:level-sets,gkps:dep-round,kmps:unified-sched} is to carefully
introduce \emph{dependencies} into the rounding process: in particular, some
sums of random variables are held fixed with probability one, while still
retaining randomness in the individual variables and guaranteeing certain
types of negative-correlation properties among them. See \cite{AS} for a
related deterministic approach that precedes these works. These
dependent-rounding approaches lead to numerous
improved approximation algorithms in scheduling, packet-routing and in several other problems 
in combinatorial optimization \cite{AS,srin:level-sets,gkps:dep-round,kmps:unified-sched, DBLP:conf/focs/ChekuriVZ10, chekuri:soda11}.

Iterative-relaxation methods based on Jain's seminal work on iterative rounding \cite{jain:comb} have been another key area of active research. In this framework, the rounding starts by computing a basic feasible solution of an LP relaxation; once the constraint matrix has a unique solution, some constraints are dropped to relax the polytope and the LP is re-solved. 
(Jain's work was in the realm of iterative \emph{approximation} algorithms; such rank-based arguments have been used in
other combinatorial-optimization contexts earlier; see, e.g., \cite{klrtvv}.) This method has been successfully used in network design, leading to fundamental works on minimum bounded degree steiner survivable network design  \cite{fleischer:01,singh:stoc07,conf/stoc/LauNSS07} and other combinatorial optimization problems
\cite{DBLP:journals/mp/BansalKKNP13,Zenklusen:2012}. See \cite{lau-ravi-singh:book} for a
comprehensive coverage of this general approach. 

We generalize the methods of dependent rounding and iterative relaxation, via a type of random walk
toward a vertex of the underlying polytope that we outline next. 
We then
present several applications in scheduling and bipartite matching through
problem-specific specializations of this approach, as well as new
concentration bounds.

The rounding approaches of \cite{klrtvv,AS,srin:level-sets,gkps:dep-round}
are generalized to linear systems as follows
in \cite{kmps:unified-sched}. Suppose we have an $n$-dimensional constraint system
$Ax \leq b$ with the additional constraints that $x \in [0,1]^n$. This will
often be an LP-relaxation, which we
aim to round to some $y \in \{0,1\}^n$ such that certain constraints in
``$Ay \leq b$'' hold with probability one, while the rest are violated ``a
little'' (with high probability). Given some feasible
$x \in [0,1]^n$, the rounding approach of \cite{kmps:unified-sched} is as
follows. First,
we assume without loss of generality that $x \in (0,1)^n$: those $x_j$ that
get rounded to $0$ or $1$ at some point, are held fixed from then on. Next,
we ``judiciously'' drop some of the constraints in ``$Ax \leq b$'' until the
number of constraints becomes smaller than $n$, thus making the system
linearly-dependent -- leading to the efficient computation of an
$r \in \Re^n$ that is in the nullspace of this reduced system. We then compute
positive scalars $\alpha$ and $\beta$ such that $x_1 : = x + \alpha r$ and
$x_2 : = x - \beta r$ both lie in $[0,1]^n$, and both
have at least one component lying in $\{0,1\}$; we then update $x$ to a
random $Y$ as: $Y := x_1$ with probability $\beta/(\alpha + \beta)$, and
$Y := x_2$ with the complementary probability $\alpha/(\alpha + \beta)$. Thus
we have rounded at least one further component of $x$, and also have the
useful property that for all $j$, $\E[Y_j] = x_j$. Different ways of
conducting the ``judicious'' reduction lead to a variety of
improved scheduling algorithms in \cite{kmps:unified-sched}.
The setting of \cite{srin:level-sets,gkps:dep-round} on
bipartite $b$-matchings can be interpreted in this framework.

\subsection{Our contributions}
\label{sec:contribs}
We further generalize the above-sketched approach of
\cite{kmps:unified-sched}. Suppose we are given a polytope $\mathcal{P}$ in
$n$ dimensions, and a \emph{non-vertex} point $x$ belonging to $\mathcal{P}$.
An appropriate basic-feasible solution will of course lead us to a vertex
of $\mathcal{P}$, but we approach (not necessarily reach) a
vertex of $\mathcal{P}$ by a random walk as follows.
Let $\mathcal{C}$ denote the set of constraints defining $\mathcal{P}$
which are satisfied \emph{tightly} (i.e., with equality) by $x$. Then,
note that there is a non-empty linear subspace $S$ of $\Re^n$ such that for
any nonzero $r \in S$, we can travel up to some strictly-positive distance
$f(r)$ along $r$ starting
from $x$, while staying in $\mathcal{P}$ and
\emph{continuing} to satisfy all
constraints in $\mathcal{C}$ tightly.
Our broad approach to conduct a
random move $Y := x + R$ by choosing an appropriately random $R$ from $S$,
such that the property ``$\E[Y_j] = x_j$'' of the previous paragraph still
holds. In particular, let $\textbf{RandMove}(x,\mathcal{P})$ -- or
simply $\textbf{RandMove}(x)$ if $\mathcal{P}$ is understood --
be as follows. Choose a nonzero $r \in S$ arbitrarily, and set
$Y := x + f(r) r$ with probability $\frac{f(-r)}{f(r) + f(-r)}$, and
$Y := x - f(-r) r$ with the complementary probability of $\frac{f(r)}{f(r) + f(-r)}$.
Note that if we repeat $\textbf{RandMove}$, we obtain a random walk
that finally
leads us to a vertex of $\mathcal{P}$; the high-level idea is to intersperse
this walk with the idea of ``judiciously dropping some constraints'' from
the previous paragraph, as well as combining certain constraints together
into one.
Three major differences from \cite{kmps:unified-sched} are: 
\begin{description}
\item[(a)] the
care given to the tight constraints $\mathcal{C}$ (\cite{kmps:unified-sched} counts the total number of constraints, and
does not exploit tightness); 
\item[(b)] the choice of which
constraint to drop being based on $\mathcal{C}$ (in \cite{kmps:unified-sched}, it is solely based on $\mathcal{P}$); and 
\item[(c)] modifying some constraints on the fly (as in steps \textbf{(D2)} and \textbf{(Modified D2)}) in 
Section~\ref{sec:sched-cap}) and
combining some
constraints into one (as in steps 2 and 3 of Algorithm \textbf{Sched-Outlier} in Section~\ref{sec:sched-outlier}).
\end{description}
The process can also be thought of as a \emph{randomized iterative relaxation} method where randomized steps are taken on an iteratively relaxed polytope. Randomization also enables concentration bounds to be employed for further analysis.


As discussed starting with Section~\ref{sec:app-capacity}, this recipe appears fruitful
in a number
of directions in scheduling, and as a new rounding technique in general.

To motivate many of our applications, we now recall a fundamental
scheduling model that has spurred many advances and applications in
combinatorial optimization, including linear-, quadratic- \&
convex-programming relaxations and new rounding
approaches \cite{LST,shmoys-tardos:gap,skutella:wjCj,azar-epstein:lp,ebenlendr:soda08,kmps:unified-sched,Gupta,bansal:stoc06,bateni:stoc09,julia:focs09}.
This model, \emph{scheduling with unrelated parallel
machines} (UPM) -- and its relatives -- play a key role in this work.
Herein, we are given a set $J$ of $n$ jobs, a set $M$ of
$m$ machines, and non-negative values $p_{i,j}$ ($i \in M, ~
j \in J$): each job $j$ has to be assigned to some machine, and
assigning it to machine $i$ will impose a processing time of $p_{i,j}$ on
machine $i$. (The word ``unrelated'' arises from the fact that there may be
no pattern among the given numbers $p_{i,j}$.) Variants such as the type of
objective function(s) to be optimized in such an assignment, whether there is
an additional ``cost-function'', whether a few
jobs can be dropped, and situations where there are release dates for,
and precedence constraints among, the jobs, lead to a rich spectrum of
problems and techniques. We now briefly discuss two such highly-impactful
results \cite{LST,shmoys-tardos:gap}. The primary UPM objective in these works
is to minimize the \emph{makespan} -- the maximum total load on any machine.
It is shown in \cite{LST} that this problem can be approximated to within a
factor of $2$; furthermore, even the natural ``restricted assignment" special case
cannot be approximated better
than $1.5$ unless $P = NP$ \cite{LST}. Despite much effort, these bounds
have not been improved; it has been shown that the \emph{value} of the objective function for the
special case of restricted assignment can be approximated to within $33/17 + \epsilon$ \cite{svensson:santa}.
The work of \cite{shmoys-tardos:gap} builds on
the upper-bound of \cite{LST} to consider the \emph{generalized assignment
problem} (GAP) where we incur a cost $c_{i,j}$ if we schedule job
$j$ on machine $i$; a simultaneous $(2,1)$--approximation for the
(makespan, total cost)-pair is developed in \cite{shmoys-tardos:gap}, leading
to numerous applications (see, e.g.,
\cite{DBLP:conf/spaa/AndreevMMS03,DBLP:journals/siamcomp/ChekuriK05}).

\subsubsection{Capacity constraints on machines}
\label{sec:app-capacity}
Handling ``hard capacities'' -- those that cannot be violated -- is
generally tricky in various settings, including facility-location and other
covering problems \cite{chuzhoy-naor:hard-cap,ghkks:hard-cap,PTW01}.
Motivated by problems in crew-scheduling
\cite{maximilian:trans85,Rushmeier:95} and by the fact that servers have a
limit on how many jobs can be assigned to them, the natural question of
scheduling with a hard capacity-constraint of ``at most $b_i$ jobs to be
scheduled on each machine $i$'' has been studied in
\cite{tsai:siam92,Zhang:01,yang:dis03,gerhard:dis05,zhang:faw09}.
The work of \cite{zhang:faw09} has shown that this problem can
be approximated to within a factor of $3$ in the special case where the
machines are \emph{identical} (job $j$ has processing time $p_j$ on any
machine). In \S~\ref{sec:sched-cap}, we use our
random-walk approach to generalize this to the setting of GAP and
obtain the GAP bounds of \cite{shmoys-tardos:gap} -- i.e.,
approximation ratios of $2$ and $1$ for the makespan and cost
respectively, while satisfying the capacity constraints: the improvements
are in the more-general scheduling
model, handling the cost constraint, and in the approximation ratio.\footnote{As described in \S~\ref{sec:sched-cap}, a referee has presented a much simpler proof of this result. We present this as well as our original proof, in the hope that perhaps the original proof has aspects that could be useful elsewhere.}
We anticipate that such a capacity-sensitive generalization of
\cite{shmoys-tardos:gap} would lead to improved approximation algorithms
for several applications
of GAP, and present one such in Section \ref{section:spaa}. However, as pointed out next in Section~\ref{sec:intro-tail},
the referee has pointed out that Theorem~\ref{thm:sched-cap} -- this capacity-sensitive generalization -- follows from the work of \cite{shmoys-tardos:gap}. 

\subsubsection{Random matchings with
sharp tail bounds}
\label{sec:intro-tail}
We obtain two types of concentration results for random matchings, as follows.

First, Theorem~\ref{thm:matching} generalizes capacitated problems (as described in the previous
application) to random bipartite
$b$-matchings with target degree bounds and sharp tail bounds for
given linear functions; see \cite{ekmsw:soda04} for applications to models
for complex networks. (Recall that given a vector $b$, a
$b$-matching is a subgraph in which every vertex $v$ has degree
at most $b(v)$.)
Given a \emph{fractional} $b$-matching $x$ in a bipartite graph $G = (J,M,E)$
of $N$ vertices, Theorem~\ref{thm:matching} shows that if there is one
linear objective function $f_i$ with bounded coefficients
associated with each $i \in M$, then we can construct (random) $b$-matchings $X$ with  
all the $|f_i(X) - f_i(x)|$ bounded independent of $N$. There has been much related work on such problems. 
For instance, given a collection of $k$ linear functions $\{f_i\}$ of $x$,
many works have considered the problem of constructing
$b$-matchings $X$ such that $f_i(X)$ is ``close'' to $f_i(x)$
simultaneously for each $i$ \cite{arora:newrounding,esa09:ravi,papad:focs00,gkps:dep-round}. The works \cite{esa09:ravi,papad:focs00} focus on the case of
constant $k$; those of \cite{arora:newrounding,gkps:dep-round} consider
general $k$, and require the usual additive ``discrepancy'' term
of $\Omega(\log N + \sqrt{f_i(x) \log N})$ in $|f_i(X) - f_i(x)|$ for most/all $i$;
in a few cases, $o(N)$ vertices will have to remain unmatched also. The work of
\cite{chekuri:soda11} considers such problems in the more-general context of matroid intersection, and
achieves the additive discrepancy term
of $O(\log N + \sqrt{f_i(x) \log N})$ in $|f_i(X) - f_i(x)|$ for all $i$. 

\smallskip \noindent \textbf{Shorter proofs derivable from earlier work.} 
It has been pointed out to us by the referee mentioned above that Theorem~\ref{thm:matching} is actually 
derivable from the work of
\cite{shmoys-tardos:gap}. This is the case when all the $r_j$ equal $1$: the only modification to be made to the algorithm underlying Theorem 2.1 of
\cite{shmoys-tardos:gap} is to write the fractional solution $x$ (actually, the related vector $x'$ in the terminology of \cite{shmoys-tardos:gap}) as
a convex combination of $b$-matchings, and choose a random $b$-matching by setting the probability of
each matching to be its coefficient in this convex combination. For the general case where the $r_j$'s are arbitrary positive integers, a little more work, as pointed out by the referee, yields 
Theorem~\ref{thm:matching} in its full generality. 
We have included the referee's elegant proof in  
\S~\ref{sec:sched-cap}. 
The referee has also pointed out that two of our consequences of Theorem~\ref{thm:matching} -- 
Theorem~\ref{thm:sched-cap} and Theorem~\ref{thm:santa2} -- follow from the work of
\cite{shmoys-tardos:gap}. We have kept these three theorems in this paper since
they give a unified approach to our applications; we thank the referee for their shorter proofs and connections to  \cite{shmoys-tardos:gap}. 
It has also been pointed out to us by Mohit Singh that our \emph{upper} bounds on
all the values $f_i(X) - f_i(x)$ (note that we have these quantities without the absolute value here), can also be obtained by the iterative-rounding
methodology as developed in \cite{lau-ravi-singh:book}; we obtain upper-bounds on the 
values $|f_i(X) - f_i(x)|$. We thank Mohit Singh for his input as well. 

Our second contribution to random matchings is a  new concentration-of-measure result; we start with informal
background first and then give some of the technical background. 
The main construction of \cite{gkps:dep-round}, which is a probabilistic analog of that of
\cite{AS}, is as follows. Suppose we have a bipartite graph $G = (U,V,E)$ with a non-negative weight $x_e$ on each edge $e$; let $s_i$ denote the sum of the weights of the edges incident on vertex $i$. Then, an efficient randomized algorithm to round each $x_e$ to a random variable $X_e \in
\{\lfloor x_e \rfloor, \lceil x_e \rceil\}$ is developed in \cite{gkps:dep-round}. Letting $S_i$ be the random variable denoting the sum of the $X_e$ over all the edges $e$ incident on $i$, this rounding algorithm has the following three properties: 
\textbf{(P1)} $\expect{X_e} = x_e$ for all $e$; \textbf{(P2)} With probability one,
$S_i \in \{\lfloor s_i \rfloor, \lceil s_i \rceil\}$ for all $i$, and \textbf{(P3)} for each vertex $i$, the random variables $((X_e - \lfloor x_e \rfloor):~e \textrm{ incident on $i$})$ are ``negatively correlated" in a natural sense, formalized by
Definition~\ref{defn:neg-correl}. These properties lead to a variety of applications in approximation algorithms \cite{AS,gkps:dep-round}. Section~\ref{sec:max-min} extends these to new concentration-of-measure bounds.
To get a feel for these, suppose  $s_i \leq 1$ for all $i$; then, we are constructing a random matching. Our bound 
in Section~\ref{sec:max-min} shows that for any subset of vertices $W$ such that $W \subseteq U$ or $W \subseteq V$, the number of vertices in $W$ that get matched, is sharply concentrated around its mean $\sum_{i \in W} s_i$. We anticipate that such bounds will be useful elsewhere as well. More formally, our contribution relies on negative correlation:

\begin{definition}[Negative Correlation for Indicator Random Variables]
\label{defn:neg-correl}
A collection of indicator random variables $\{Z_{i}\}, i \in [1,N]$ is said to be negatively correlated if for any $t$, any
$1 \leq i_1 < i_2 < \cdots < i_t \leq N$, and any $b \in \{0,1\}$, $\prob{\bigwedge_{j=1}^{t} (Z_{i_j}=b)} \leq \prod_{j=1}^{t} \prob{Z_{i_j}=b}$.
\end{definition}

One of the key benefits of negative correlation is that such a ``self-correcting'' property leads to strong concentration:

\begin{theorem}
\label{thm:chbound-negcorrel}
{\bf (The Chernoff-Hoeffding bound under negative correlation \cite{ps:edge-col}):} Suppose $X=\sum_{i} X_{i}$ where $X_{i}$ are negatively correlated random
variables taking values in $\{0,1\}$. Then:
\begin{description}
\item[(i)] if $\expect{X} \geq \mu$ and $\delta \in [0,1]$, then $\prob{ X \leq \mu(1-\delta)} \leq e^{-\mu\delta^2/2}$;
\item[(ii)] if $\expect{X} \leq \mu$ and $\delta \geq 0$, then $\prob{ X \geq \mu(1+\delta)} \leq e^{-\mu[(\delta+1)\ln{(\delta+1)}-\delta]}$.
\end{description}
\end{theorem}

A natural question one can ask is whether the negative-correlation property (P3) of \cite{gkps:dep-round} does not just hold ``locally'' (at a vertex $i$), but across the graph as well. Unfortunately, it is easy to show that such a property fails to hold: in fact, by a large margin.\footnote{Suppose, e.g., that $U = \{u_1, u_2\}$ and
$V = \{v_1, v_2\}$ and that we have the complete bipartite graph on $(U,V)$, with $x_e = 1/2$ for each of the four edges
$e$. Then, the only solution here is to select $\{(u_1, v_1), (u_2, v_2)\}$ with probability $1/2$, and
$\{(u_1, v_2), (u_2, v_1)\}$ with the remaining probability of $1/2$. Thus, the edges $(u_1, v_1)$ and $(u_2, v_2)$ are
perfectly positively correlated.} However, we are able to show in Theorem~\ref{theorem:neg1} that such negative correlation holds if we only consider any collection of vertices on the ``same side'' of $G$: all in $U$ or all in $V$. 
(Theorem~\ref{theorem:neg1} is stated in the context of matchings, wherein $s_i \leq 1$ for all $i$, but its proof directly generalizes to arbitrary $s_i$.) The resultant concentration inequalities that follow from Theorem~\ref{thm:chbound-negcorrel} are crucially needed in Section~\ref{sec:max-min}. 

\subsubsection{Scheduling with outliers: makespan and fairness}
Note that the $(2,1)$ bicriteria approximation that we obtain for GAP as described in Section~\ref{sec:app-capacity},
generalizes the results of \cite{shmoys-tardos:gap}. We now present such a
generalization in another direction: that of ``outliers'' in scheduling
\cite{Gupta}. For instance, suppose that
in the ``processing times $p_{i,j}$ and
costs $c_{i,j}$'' setting of GAP,
we also have a profit $\pi_j$ for choosing to schedule each job $j$.
Given a ``hard'' target profit $\Pi$, target makespan $T$ and total cost
$C$, the LP-rounding method of \cite{Gupta} either proves that these
targets are not simultaneously achievable, or constructs a schedule
with values $(\Pi, 3T, C(1 + \epsilon))$ for any constant $\epsilon > 0$.
Using our rounding approach, we improve
this to $(\Pi, (2 + \epsilon)T, C(1 + \epsilon))$ in
\S~\ref{sec:sched-outlier}.
(The factors of $\epsilon$ in the cost are required due to the hardness of
knapsack \cite{Gupta}.)
Also, fairness is a fundamental
issue in dealing with outliers:
e.g., in repeated runs of such algorithms, we may not desire some jobs $j$ being excluded as an outlier in almost all of the repetitions just so that the global objective function remains high. 
Theorem~\ref{thm:sched-fairness} accommodates
fairness in the form of
scheduling-probabilities for the jobs that can be part of the
input.

\subsubsection{Max-Min Fair Allocation}
This problem is the
max-min version of UPM, where we aim to maximize the minimum
``load'' (viewed as utility) on the machines; it has received a good
deal of attention \cite{bansal:stoc06,asadpour:stoc07,feige:soda08,asadpour-feige-saberi,bateni:stoc09,julia:focs09}. We obtain a new algorithm for 
max-min fair allocation to near-optimally determine the integrality
gap of a well-studied ``configuration LP" relaxation via bipartite dependent rounding and its generalization \cite{gkps:dep-round}. (Also, the results of  \cite{shmoys-tardos:gap} imply a generalization of a result of \cite{dani:05} on max-min fairness
to the setting of
equitable partitioning of the jobs; see Theorem~\ref{thm:santa2}.) Improved approximation factors are now known due to Chakrabarty, Chuzhoy and Khanna \cite{Chuzhoy09} and Bateni, Charikar and Guruswamy \cite{bateni:stoc09}, via an approach that avoids the configuration LP.

\subsubsection{Overlay Networks for Streaming}
As an additional application, Section~\ref{section:spaa} improves upon some of the rounding techniques of
\cite{DBLP:conf/spaa/AndreevMMS03} in the design of overlay  networks for streaming.


\section{Random Matchings with Linear Constraints, and GAP with Capacity Constraints}
\label{sec:sched-cap}

We develop an efficient scheme to generate random
subgraphs of bipartite graphs that
satisfy hard degree-constraints and near-optimally satisfy
a collection of linear constraints; this is captured by Theorem~\ref{thm:matching}. As mentioned in the introduction, a referee has pointed out that there is a much-shorter proof for Theorem~\ref{thm:matching} that is motivated by an approach of \cite{shmoys-tardos:gap}. We give this short proof due to the referee below, and also keep our original argument for completeness. We start by defining the
input for the algorithm that is guaranteed by Theorem~\ref{thm:matching}:

\begin{definition}
\label{defn:bip-degree-linear}
(Matchings with Structured Linear Constraints (MSLC))
The input to the MSLC problem consists of the following:
\begin{itemize}
\item a bipartite graph $G = (J, M, E)$   with ``jobs'' $J$ and
``machines'' $M$; let $\mathcal{F}$ be the collection of
edge-indexed vectors $y$ (with $y_{i,j}$ denoting $y_e$ where $e =
(i,j) \in E$). 
\item an integer \emph{requirement}
$r_j$ for each $j \in J$ and an integer \emph{capacity} $b_i$ for each
$i \in M$.
\item for each $i \in M$, a linear objective function
$f_i: \mathcal{F} \rightarrow \Re$ given by $f_i(y) = \sum_{j:~(i,j) \in E}
p_{i,j} y_{i,j}$ such that $0 \leq p_{i,j} \leq \ell_i$ for each $j$, where $\ell_i$ is some given positive
value associated with $i$. 
\item a global cost constraint $\sum_{i,j}c_{i,j}y_{i,j} \leq C$, and
\item a vector $x \in \mathcal{F}$ with $x_e \in [0,1]$ for each $e$ which satisfies the given constraints; i.e., we have
(i) $\sum_i x_{i,j} \geq r_j$ for each $j$, (ii) $\sum_j x_{i,j} \leq b_i$ for each $i$, and
(iii)  $\sum_{i,j} c_{i,j} x_{i,j} \leq C$. 
\end{itemize}
\end{definition}

\begin{theorem}
\label{thm:matching}
(Rounding MSLC instances)
Suppose we are given an instance of MSLC with parameters as in Definition~\ref{defn:bip-degree-linear}.
Then, we can efficiently construct a random subgraph of $G$ given by a
binary vector $X \in \mathcal{F}$, such that:
(a) with probability one, each $j \in J$ has degree at least $r_j$,
each $i \in M$ has degree at most $b_i$, and
$|f_i(X) - f_i(x)| < \ell_i~\forall i$; as well as
(b) for all $e \in E$, $\expect{X_e} = x_e$ which implies
$\expect{\sum_{i,j}c_{e}X_{e}}=\sum_{e}c_{e}x_{e}= C$.
\end{theorem}

\smallskip \noindent \textbf{An elegant and short proof of Theorem~\ref{thm:matching} due to a referee.} The referee's short proof is as follows. Given a vector $x$ as in MSLC and for each machine $i$, define a permutation $\pi_i(1), \pi_i(2), \ldots, \pi_i(n)$ of the jobs (where $n = |J|$) such that 
\[ p_{i, \pi_i(1)} \geq  p_{i, \pi_i(2)} \geq  \cdots \geq p_{i, \pi_i(n)}. \]
We write an alternative system of inequalities for which $x$ is clearly a feasible solution:
\begin{eqnarray}
\sum_i x_{i,j} & \geq & r_j ~~\forall j; \label{eqn:alt-job} \\
\sum_{k = 1}^{\ell} x_{i,\pi_i(k)} & \leq & \left \lceil \sum_{k = 1}^{\ell} x_{i,\pi_i(k)} \right \rceil ~~\forall i ~\forall \ell \leq n;  \label{eqn:alt-machine-upper} \\
\sum_{k = 1}^{\ell} x_{i,\pi_i(k)} & \geq & \left \lfloor \sum_{k = 1}^{\ell} x_{i,\pi_i(k)} \right \rfloor ~~~\forall i ~ \forall \ell \leq n;  \label{eqn:alt-machine-lower} \\
x_{i,j} & \in & [0,1] ~~\forall (i,j). \label{eqn:alt-0-1-bounded}
\end{eqnarray}
The claim is that the polytope induced by (\ref{eqn:alt-job}), (\ref{eqn:alt-machine-upper}), 
(\ref{eqn:alt-machine-lower}), and (\ref{eqn:alt-0-1-bounded}) is integral; the proof follows the usual method for showing the integrality of the matroid-intersection polytope, and is as follows. Consider the tight constraints at any extreme point $z$, and view the constraints as only on those those $z_{i,j}$ that are yet-unrounded (i.e., lie in $(0,1)$); in other words, we view the $z_{i,j} \in \{0,1\}$ as fixed, and not as variables any more. 
Suppose there are $v$ yet-unrounded $z_{i,j}$'s, where $v \not= 0$ for a contradiction. The tight constraints corresponding to (\ref{eqn:alt-0-1-bounded}) constitute a partition matroid and are easily-seen to be at most $v/2$ in number; similarly, the tight constraints corresponding to (\ref{eqn:alt-machine-upper}) and 
(\ref{eqn:alt-machine-lower}) yield a laminar system and are also at most $v/2$ in number. Further, if both of these ``$v/2$'' bounds are tight, then these two systems of tight constraints span the constraint ``$\sum_i \sum_j z_{i,j} = \sum_i \sum_j x_{i,j}$''. Thus, the tight constraints do not span all of a $v$-dimensional space, a contradiction. Hence our polytope is integral. 
Now, as usual, all we need is to decompose $x$ as a convex combination of vertices of this integral polytope, and pick a random vertex of the polytope: for each vertex, its probability equals its coefficient in the convex combination. The fact that (\ref{eqn:alt-machine-upper}) easily helps show that for the obtained random binary vector $X$, $f_i(X) - f_i(x) < \ell_i$ for any $i$; similarly, (\ref{eqn:alt-machine-lower}) implies that $f_i(X) - f_i(x) > -\ell_i$. 

This completes the description of the referee's elegant and short proof; we now return
 to our approach. 

\smallskip  
We  first prove an important special case of Theorem~\ref{thm:matching}:
GAP with individual capacity constraints on each machine. This special case --
handled by Theorem~\ref{thm:sched-cap} --
captures much of the essence of
Theorem~\ref{thm:matching}; the full proof of Theorem~\ref{thm:matching}
follows after Theorem~\ref{thm:sched-cap}. This is a special case in the following senses:
\begin{itemize}
\item $r_j = 1$ for each $j$;
\item we only require that the rounding $X$ is such that for each $i$, $f_i(X)$ is not ``much more" than $f_i(x)$: i.e., the $f_i(X)$  values are allowed to be much smaller than the corresponding $f_i(x)$ values; and
\item for some $T$, all the $\ell_i$ and $f_i(x)$ are upper-bounded by $T$.
\end{itemize}
In words, this is the UPM problem with hard capacities $b_i$ on the machines as discussed in Section~\ref{sec:app-capacity}. 
(As mentioned at the beginning of our description of Algorithm \textbf{Sched-Cap}, 
we guess the optimum makespan $T$ by binary search -- and if the processing time $p_{i,j}$ is strictly larger than $T$, we set the corresponding decision
variable $x_{i,j}$ to $0$.)   
Note, as pointed out in Section~\ref{sec:intro-tail}, that Theorem~\ref{thm:sched-cap} can also be derived from
the work of \cite{shmoys-tardos:gap}.

Our main contribution here is an efficient algorithm \textbf{Sched-Cap} that
has the following guarantee, generalizing the GAP bounds of
\cite{shmoys-tardos:gap}:
\begin{theorem}
\label{thm:sched-cap}
There is an efficient algorithm \textbf{Sched-Cap} that returns a schedule maintaining all the capacity constraints, of cost at most $C$ and makespan
at most $2T$, where $T$ is the
optimal makespan with cost $C$ that satisfies the capacity constraints.
\end{theorem}

\paragraph{{\bf Algorithm \textbf{Sched-Cap}}}
Algorithm \textbf{Sched-Cap} proceeds as follows. First we guess the optimum makespan $T$ by binary search as in \cite{LST}.
If $p_{i,j} >T$, $x_{i,j}$ is set to $0$. The solution to the following
integer program gives the optimum schedule:
\begin{eqnarray*}
&&\sum_{i,j} c_{i,j}x_{i,j} \leq C ~~~~~~~~~~~~~~~~~~~~ (\text{Cost})\\
&&\sum_{i,j} x_{i,j} =1 ~\forall j ~~~~~~~~~~~~~~~~~~~ (\text{Assign}) \\
&& \sum_{j} p_{i,j}x_{i,j} \leq T ~\forall i  ~~~~~~~~~~~~~~~~~ (\text{Load})\\
&& \sum_{j} x_{i,j} \leq b_{i} ~\forall i   ~~~~~~~~~~~~~~~~~ (\text{Capacity})\\
&& x_{i,j} \in \{0,1\} ~\forall i,j  ~~~~~~~~~~~~~~~~~ (\text{Binary})          \\
&& x_{i,j}=0~~\text{ if } p_{i,j} >T ~~~~~~~~~~~~~ (\text{Filtering})  
\end{eqnarray*}

We relax the constraint ``$x_{i,j} \in \{0,1\} ~\forall (i, j)$'' to
``$x_{i,j} \in [0,1] ~\forall (i,j)$'' to obtain an LP relaxation
\textbf{LP-Cap}. We solve the LP to obtain an optimal LP solution $x^{*}$;
we next show how \textbf{Sched-Cap} rounds $x^{*}$ using algorithm \textbf{RandMove} of Section~\ref{sec:contribs}  to obtain a good integral solution. 

\smallskip \noindent \textbf{Remark: dropping the cost constraint.} Although we mention the constraint (Cost) above
for completeness, we will \emph{drop} this constraint from now on. This is because -- as shown in the next
paragraph -- our final rounded vector $X$ will satisfy $\expect{X_{i,j}} =x^{*}_{i,j}$, and hence (Cost) will be
satisfied in expectation; all our other guarantees are with probability one. The entire process as we demonstrate at the end can be derandomized and hence the cost upper bound of $C$ is obeyed.

\smallskip
Note that $x^{*}_{i,j} \in [0,1]$ denotes the ``fraction'' of job $j$
assigned to machine $i$. Initialize $X=x^{*}$. The algorithm is composed of
several iterations. The random value of the assignment-vector $X$ at the end
of iteration $h$ of the overall algorithm is denoted by $X^{h}$. 
Each iteration $h$ conducts a randomized update using algorithm \textbf{RandMove} on the polytope
of a linear system constructed from a \emph{subset} of the constraints of
\textbf{LP-Cap}. Therefore, by induction on $h$, we will have for all $(i,j,h)$
that $\expect{X_{i,j}^{h}}=x^{*}_{i,j}$. 

Let $J$ and $M$ denote the set of jobs and machines, respectively.
Suppose we are at the beginning of some iteration $(h+1)$ of the overall algorithm: we are currently looking at the values $X_{i,j}^{h}$. We will maintain four invariants.

\SN
\textbf{Invariants across iterations:}

\begin{description}
\item[(I1)] Once a variable $x_{i,j}$ gets assigned to $0$ or $1$, it is
never changed;
\item[(I2)] The constraints (Assign) always hold; and
\item[(I3)] Once a constraint in (Capacity) becomes tight, it remains tight, and
\item[(I4)] Once a constraint is dropped in some iteration, it is never
reinstated.
\end{description}
Iteration $(h+1)$ of \textbf{Sched-Cap} consists of three main steps:

\SN 1.
We first \emph{hard-wire} all $X_{i,j}^{h} \in \{0,1\}$; thus, the variables 
$X_{i,j}^{h} \in (0,1)$ yield the current vector $Y$ of ``floating''
(to-be-rounded) variables; let $\mathcal{S} \equiv (A_h Y = u_h)$ denote
the current linear system that represents \textbf{LP-Cap}. ($A_h$ is some
matrix and $u_h$ is a vector; we avoid using ``$\mathcal{S}_h$'' to
simplify notation.) In particular, the ``capacity''
of machine $i$ in $\mathcal{S}$ is its residual capacity $b_i'$,
i.e., $b_i$ minus the number of jobs that have been permanently assigned to
$i$ thus far. Recall that the cost constraint is \emph{not} included in the constraint matrix
$A_h Y = u_h$; we continue to maintain $A_h Y = u_h$ exactly. 

\SN 2. Let $Y \in \Re^v$ for some $v$; note that $Y \in (0,1)^v$.
Let $M_k$ denote the set of all machines
$i$ for which exactly $k$ of the values $Y_{i,j}$ are positive.
We will now drop some of the constraints in $\mathcal{S}$:
\begin{description}
\item[(D1)] for each $i \in M_1$, we drop its load and capacity constraints from
$\mathcal{S}$;
\item[(D2)] for each $i \in M_2$, we drop its load constraint and rewrite its capacity
constraint as $x_{i,j_1}+x_{i,j_2} \leq \lceil X^{h}_{i,j_1}+X^{h}_{i,j_2}\rceil$, where $j_1,j_2$
are the two jobs fractionally assigned to $i$.
\item[(D3)] for each $i \in M_3$ for which
\emph{both} its load and capacity
constraints are tight in $\mathcal{S}$, we drop its load constraint from
$\mathcal{S}$.
\end{description}

\SN 3. Let $\mathcal{P}$ denote the polytope defined by this reduced system of
constraints. A key claim that is proven in Lemma~\ref{lemma:vertex} below
is that $Y$ is \emph{not} a vertex of $\mathcal{P}$. We now invoke
$\textbf{RandMove}(Y,\mathcal{P})$; this is allowable if
$Y$ is indeed not a vertex of $\mathcal{P}$.

\SN The above three steps complete iteration $(h+1)$.

\SN
\paragraph{{\bf Analysis}}

\SN It is not hard
to verify that the invariants ({\bf I1})-({\bf I4}) hold true (though the fact
that we drop the all-important capacity constraint for machines
$i \in M_1$ may look bothersome, a moment's reflection shows that such
a machine cannot have a tight capacity-constraint since its sole
relevant job $j$ has value $Y_{i,j} \in (0,1)$).
Since we
make at least one further constraint tight via
$\textbf{RandMove}$ in each iteration, invariant ({\bf I4}) shows that we
terminate, and that the number of iterations is at most the initial number
of constraints. Let us next present Lemma~\ref{lemma:vertex}, a key lemma:

\begin{lemma}
\label{lemma:vertex}
In no iteration is $Y$ a vertex of the current polytope $\mathcal{P}$.
\end{lemma}

\begin{proof}
Suppose that in a particular iteration, $Y$ \textbf{is} a vertex of
$\mathcal{P}$. Fix the notation $v$, $M_k$ etc.\ w.r.t.\ this iteration;
let $m_k = |M_k|$, and let $n'$ denote the remaining number of jobs that
are yet to be assigned permanently to a machine.
Let us lower- and upper-bound the number of variables $v$.
On the one hand, we have
\begin{equation}
\label{eqn:v-defn}
v = \sum_{k \geq 1} k \cdot m_k,
\end{equation}
by definition of the sets $M_k$; since each remaining job $j$ contributes at
least two variables (coordinates for $Y$), we also have
\begin{equation}
\label{eqn:lb-v}
v \geq 2n'. 
\end{equation}
On the other hand, since $Y$ has been assumed to be a vertex of
$\mathcal{P}$, the number $t$ of constraints in $\mathcal{P}$ that
are satisfied \emph{tightly} by $Y$, must be at least $v$. How large
can $t$ be? Each current job contributes one (Assign) constraint
to $t$; by our ``dropping constraints'' steps ({\bf D1}), ({\bf D2}) and ({\bf D3}) above,
the number of tight constraints (``load'' and/or ``capacity'')
contributed by the machines is at most $m_2 + m_3 + \sum_{k \geq 4} 2 m_k$.
Thus we have
\[ v \leq  t \leq n' + m_2 + m_3 + \sum_{k \geq 4} 2 m_k, \]
i.e., 
\begin{equation}
\label{eqn:ub-v}
2n' \geq 2v - 2 m_2 - 2 m_3 - \sum_{k \geq 4} 4 m_k. 
\end{equation}

Eliminating the term $2n'$ between (\ref{eqn:lb-v}) and (\ref{eqn:ub-v}) and then using the definition of
$v$ from (\ref{eqn:v-defn}), we get
\[ m_1 + 2 m_2 + 3 m_3 + 4 m_4 + \cdots \leq 2 m_2 + 2 m_3 + \sum_{k \geq 4} 4 m_k. \]
This is possible only if:
(i) $m_1 = m_3 = 0$ and $m_5 = m_6 = \cdots = 0$;
(ii) the capacity constraints
are tight for all machines in $M_2 \cup M_4$ -- i.e., for all machines;
and
(iii) $t = v$. However, in such a situation, the $t$ constraints in
$\mathcal{P}$
constitute the \emph{tight} assignment constraints for the jobs and
the \emph{tight} capacity constraints for the machines, and are hence
linearly dependent (since the total assignment ``emanating from'' the
jobs must equal the total assignment ``arriving into'' the machines).
Thus we reach a contradiction, and hence $Y$ is not a vertex of
$\mathcal{P}$.
\end{proof}

We next show that the final makespan is at most $2T$ with probability one:

\begin{lemma}
\label{lem:sched-cap-span}
Let $X$ denote the final rounded vector. Algorithm \textbf{Sched-Cap} returns
a schedule, where with probability one: (i) all capacity-constraints
on the machines are satisfied, and (ii) for all $i$,
$\sum_{j \in J} X_{i.j}p_{i,j} < \sum_{j}x_{i,j}^{*}p_{i,j} +
\text{max}_{j \in J:~x^{*}_{i,j}\in (0,1)}p_{i,j}$.
\end{lemma}

\begin{proof}
For part (i), the only care to be taken is for machines $i$ that end up
in $M_1$ and hence have their capacity-constraint dropped. However,
as argued soon after the description of the three steps of an iteration,
note that such a machine cannot have a tight capacity-constraint when
this constraint was dropped; hence, even if the remaining job $j$ got
assigned finally to $i$, its capacity constraint cannot be violated.

Let us now prove (ii). Fix a machine $i$. If at all its load-constraint
was dropped, it must be when $i$ ended up in $M_1, M_2$ or $M_3$. The
case of $M_1$ is argued as in the previous paragraph. So suppose
$i \in M_{\lambda}$ for some $\lambda\in \{2,3\}$ when its load constraint got
dropped.  Let us first consider
the case $\lambda= 2$. Let the two jobs fractionally assigned on $i$ at that point
have processing times $(p_1,p_2)$ and fractional assignments $(y_1, y_2)$
on $i$, where $0 \leq p_1, p_2 \leq T$, $\max\{p_1, p_2\} > 0$, and $0 < y_1, y_2 < 1$.
If $y_1+y_2 \leq 1$, we know that at the end, the
assignment vector $X$ will have at most one of $X_1$ and $X_2$ being one.
Then, $p_1 X_1 + p_2 X_2 \leq \max\{p_1, p_2\} < p_1 y_1 + p_2 y_2 + \max\{p_1, p_2\}$ as required.
If $1< y_1+y_2 \leq 2$, then both $X_1$ and $X_2$ can be assigned and again,
$p_1 X_1 + p_2 X_2 < p_1 y_1 + p_2 y_2 + \max\{p_1, p_2\}$. For the case $\lambda=3$,
we know from ({\bf I3}) and ({\bf D3}) that its capacity-constraint must
be \emph{tight} at some integral value $u$ at that point, and that this
capacity-constraint was preserved until the end.  We must have $u = 1$
or $2$ here. Let us just
consider the case $u = 2$; the case of $u = 1$ is similar to the case of $\lambda = 2$ with $y_1 + y_2 \leq 1$.
Here again, simple
algebra yields
that if $0 \leq p_1, p_2, p_3 \leq T$ and $0 < y_1, y_2, y_3 < 1$
with $y_1 + y_2 + y_3 = u = 2$, then for any binary vector $(X_1, X_2, X_3)$ of
Hamming weight $u = 2$,
$p_1 X_1 + p_2 X_2 + p_3 X_3 < p_1 y_1 + p_2 y_2 + p_3 y_3 +
\max\{p_1, p_2, p_3\}$.
\end{proof}

Finally we have the following lemma.

\begin{lemma}
\label{lem:sched-cap-cost}
Algorithm \textbf{Sched-Cap} can be derandomized to create a schedule of cost
at most $C$.
\end{lemma}

\begin{proof}
(Sketch)
Let $X_{i,j}^{h}$ denote the value of $x_{i,j}$ at iteration $h$. We know for all $i,j,h$,  $E[X^{h}_{i,j}]=x^{*}_{i,j}$, where $x^{*}_{i,j}$ is solution of \textbf{LP-Cap}. Therefore, at the end, we have that the total expected cost
incurred is $C$. The procedure can be derandomized directly by the method of
conditional expectations, giving an $1$-approximation to the cost.
\end{proof}

Lemmas \ref{lem:sched-cap-span} and \ref{lem:sched-cap-cost} yield
Theorem~\ref{thm:sched-cap}.

\smallskip We next turn to the proof of Theorem \ref{thm:matching}. The key difference from Theorem~\ref{thm:sched-cap}
is that constraint (Load') must now be approximated well \emph{both} from above and below by our rounding, as opposed to just bounding the deviation above. A less-critical difference is in (Assign'), where the $r_j$'s can be an arbitrary positive integers intead of $1$. 

\paragraph{{\bf Proof of Theorem \ref{thm:matching}}}

We now consider the full proof of Theorem \ref{thm:matching}. The following
integer program gives an optimal matching:

\begin{eqnarray*}
&&\sum_{i,j} c_{i,j}x_{i,j} \leq C ~~~~~~~~~~~~~~~~~~~~~ (\text{Cost'})\\
&&\sum_{i,j} x_{i,j} \geq r_j ~\forall j ~~~~~~~~~~~~~~~~~~~ (\text{Assign'}) \\
&& \sum_{j} p_{i,j}x_{i,j} = f_i ~\forall i  ~~~~~~~~~~~~~~~~~ (\text{Load'})\\
&& \sum_{j} x_{i,j} \leq b_{i} ~\forall i   ~~~~~~~~~~~~~~~~~ (\text{Capacity'})\\
&& x_{i,j} \in \{0,1\} ~\forall i,j\\
&& x_{i,j}=0~~\text{ if } p_{i,j} >\ell_i
\end{eqnarray*}

The proof of Theorem \ref{thm:matching} is quite similar to Theorem \ref{thm:sched-cap}.
We elaborate upon the necessary modifications. First, while hard-wiring those $X_{i,j}^{h} \in \{0,1\}$ and viewing
the current linear system as having only those $X_{i,j}^{h} \in (0,1)$ as variables, 
we update the assignment requirements of the jobs as well as the capacity constraints of the machines
accordingly. (That is, we subtract the contributions of the variables $X_{i,j}^{h} \in \{0,1\}$ to obtain the residual
demands $r_j$ and the residual capacities $b_i$.)  The dropping rules ({\bf D1}) and ({\bf D3})
remain the same. However, ({\bf D2}) is modified as follows:

({\bf Modified D2})  {\it For each $i \in M_2$, we drop its load constraint and rewrite its capacity constraint.
Let $j_1,j_2$ be the two jobs assigned to machine $i$ with fractional assignment $x_{i,j_1}$
and $x_{i,j_2}$. Then if $x_{i,j_1}+x_{i,j_2} \leq 1$, set the capacity constraint to
$x_{i,j_1}+x_{i,j_2} \leq 1$. Else if $1 < x_{i,j_1}+x_{i,j_2} <  2$, set the capacity constraint to
$x_{i,j_1}+x_{i,j_2} \geq 1$}.

Lemma \ref{lemma:vertex},  Lemma \ref{lem:sched-cap-cost} remain unchanged. We have a new Lemma \ref{lem:match-cap-span} corresponding to Lemma \ref{lem:sched-cap-span}, which we prove next.

\begin{lemma}
\label{lem:match-cap-span}
Let $X$ denote the final rounded vector.
Then $X$ satisfies with probability one: (i) all capacity-constraints
on the machines are satisfied, and (ii) for all $i$,
$\sum_{j}x_{i,j}^{*}p_{i,j} -
\text{max}_{j \in J:~x^{*}_{i,j}\in (0,1)}p_{i,j}< \sum_{j \in J} X_{i,j}p_{i,j} < \sum_{j}x_{i,j}^{*}p_{i,j} +
\text{max}_{j \in J:~x^{*}_{i,j}\in (0,1)}p_{i,j}$.
\end{lemma}

\begin{proof}
Part (i) is similar to Part (i) of Lemma \ref{lem:sched-cap-cost} and follows
from the facts that the capacity constraints are never violated and machines in $M_1$
cannot have tight capacity constraints.

Let us now prove (ii). Note that in ({\bf Modified D2})
the upper bound on capacity constraint is maintained as in ({\bf D2}).
Hence from Lemma \ref{lem:sched-cap-span}, we get $\sum_{j \in J} X_{i,j}p_{i,j} < \sum_{j}x_{i,j}^{*}p_{i,j} +
\text{max}_{j \in J:~x^{*}_{i,j}\in (0,1)}p_{i,j}$. So we only need to show the lower bound
on the load.
Fix a machine $i$. If at all its load-constraint was dropped, it must be when
$i$ ended up in $M_1\cup M_2 \cup M_3$. In the case of $M_1$, at most one job fractionally assigned to it
may not be assigned in the final rounded vector. So suppose $i \in M_{\lambda}$ for some $\lambda \in \{2,3\}$
when $i$ has its load constraint dropped. Let us first consider the case of $\lambda =2$. Let the two jobs
fractionally assigned to $i$ at that point
have processing times $(p_1,p_2)$ and fractional assignments $(y_1, y_2)$
on $i$, where $0 \leq p_1, p_2 \leq T$, and $0 < y_1, y_2 < 1$.
If $y_1+y_2 \leq 1$, then at the end, none of the jobs may get assigned.
Simple algebra now shows that
$0 > p_1 y_1 + p_2 y_2 - \max\{p_1, p_2\}$ as required.
If $1< y_1+y_2 \leq 2$, then at least one of the two jobs $X_1$ and $X_2$ get assigned to $i$ and again,
$p_1 X_1 + p_2 X_2 > p_1 y_1 + p_2 y_2 - \max\{p_1, p_2\}$. For the case $\lambda =3$,
we know from ({\bf I3}) and ({\bf D3}) that $i$'s capacity-constraint must
be \emph{tight} at some integral value $u$ at that point, and that this
capacity-constraint was preserved until the end.  We must have $u = 1$
or $2$ in this case. Let us just
consider the case $u = 2$; the case of $u = 1$ is similar to the case of $\lambda= 2$ with $y_1 + y_2 \leq 1$.
Here again, simple
algebra yields
that if $0 \leq p_1, p_2, p_3 \leq T$ and $0 < y_1, y_2, y_3 < 1$
with $y_1 + y_2 + y_3 = u = 2$, then for any binary vector $(X_1, X_2, X_3)$ of
Hamming weight $u = 2$,
$p_1 X_1 + p_2 X_2 + p_3 X_3 > p_1 y_1 + p_2 y_2 + p_3 y_3 -\max\{p_1, p_2, p_3\}$.
\end{proof}

Lemmas \ref{lem:match-cap-span} and \ref{lem:sched-cap-cost} yield
Theorem~\ref{thm:matching}.

This completes the present section. We have shown
how a random subgraph of a bipartite graph with hard degree-constraints
can be obtained while near-optimally satisfying a collection of linear
constraints as well as a given cost-budget. As a special case of this,
we obtained a $2$-approximation algorithm for
the generalized assignment problem
with hard capacity-constraints on the machines.

\section{Scheduling with Outliers}
\label{sec:sched-outlier}

In this section, we consider GAP with outliers
and with a hard profit constraint \cite{Gupta}. Formally, the problem is as follows. 

Suppose we are given $m$ machines and $n$ jobs, where job $j$ requires processing time of $p_{i,j}$ in machine $i$, incurs a cost of $c_{i,j}$ if assigned to $i$, and provides a profit of $\pi_{j}$ if scheduled.
Let $x_{i,j}$ be the indicator variable for job $j$ to be scheduled on machine $i$. 
The goal is to minimize the makespan $T=\text{max}_{i}\sum_{j}x_{i,j}p_{i,j}$, subject to the constraints that the total cost $\sum_{i,j}x_{i,j}c_{i,j}$ is at most $C$ and total profit $\sum_{j}\pi_{j}\sum_{i}x_{i,j}$ is at least $\Pi$. Dropping a few outliers with high processing requirement can often improve the scheduling performance substantially; however, we would like to drop as few outliers as possible. The problem formulation captures this by assigning a profit to each scheduled job, in addition to maintaining the total cost of assignment and the makespan constraints.

Our main contribution here is the following:
\begin{theorem}
\label{thm:sched-outlier}
For any given constant $\epsilon >0$, there is an efficient algorithm \textbf{Sched-Outlier} that returns a schedule of profit at least $\Pi$, cost at most $C(1+\epsilon)$ and makespan at most $(2+\epsilon)T$, where $T$ is the optimal makespan among all schedules that simultaneously have cost $C$ and profit $\Pi$.
\end{theorem}

This is an improvement over the work of Gupta, Krishnaswamy, Kumar and Segev \cite{Gupta}, where they constructed a schedule
with makespan $3T$, profit $\Pi$ and cost $C(1+\epsilon)$. In addition, our approach also accommodates \emph{fairness} -- a basic requirement in
dealing with outliers -- especially when problems have to be run repeatedly. It ensures that each job gets a fair chance of being scheduled.
We formulate fairness via a stochastic program that specifies for each job $j$,
a lower-bound $r_j$ on the probability that it gets scheduled: 

\begin{definition}[Fairness]
Given a set of jobs $J$ and a real $r_j \in [0,1]$ for all $j \in J$, a schedule is said to be fair if for every job $j$, it is assigned to a machine with probability at least $r_j$.
\end{definition}

We adapt our approach to honor such requirements:

\begin{theorem}
\label{thm:sched-fairness}
Suppose we are given a vector $r = (r_j:~j \in J) \in [0,1]^n$ of fairness requirements.  
There is an efficient randomized algorithm  that returns a schedule of expected profit at least $\Pi$, expected cost at most $2C$, makespan at most $3T$ with probabiity $1$, and guarantees that for each job $j$, it is scheduled with the specified probability
$r_j$, where $T$ is the optimal expected makespan with expected cost $C$, expected profit $\Pi$, and under the
fairness requirements $r$. 
\end{theorem}

We start with Theorem \ref{thm:sched-outlier} and describe the algorithm \textbf{Sched-Outlier} first. Next,
we prove Theorem \ref{thm:sched-fairness}. While the main ideas behind \textbf{Sched-Outlier} are similar to those of
Section~\ref{sec:sched-cap}, the choice of constraints to drop becomes more complex. It is now possible that there exist singleton jobs each with only one fractional assignment to a machine, which was not possible if all jobs need to be assigned as in Section~\ref{sec:sched-cap}. Hence, we may not be able to maintain the (Assign) constraints always. For jobs whose (Assign) constraints are dropped, we carefully maintain the total profit obtained from these jobs. This leads to a few possible configurations at a vertex of the polytope. We provide a rounding scheme for each of these configurations, leading to the desired approximation factors. We now describe the algorithm in more detail. 

\paragraph{{\bf Algorithm  Sched-Outlier}}
 The algorithm starts by guessing the optimal makespan $T$ by binary search as in \cite{LST}.
 If $p_{i,j} > T$, then $x_{i,j}$ is set to $0$. Next let $\epsilon > 0$ be the given constant. The
 running time of the algorithm depends on $\epsilon$ and is $O(n^{(1/\epsilon)^{O(1)}})$. 
 We ``guess" all assignments $(i,j)$ where $c_{i,j} > \epsilon' C$, with $\epsilon'=\epsilon^2$. Any valid schedule
can have at most $1/\epsilon'$ pairs with assignment costs higher than
$\epsilon' C$; hence, this guessing (i.e., enumeration)
can be done in time $O((mn)^{\frac{1}{\epsilon'}})=O((mn)^{1/\epsilon^2})$. For all $(i,j)$ with
$c_{i,j} > \epsilon' C$, let $\mathcal{G}_{i,j} \in \{0,1\}$ be a correct
guessed assignment: by our polynomial-time enumeration, we may thus assume we know the optimal $\mathcal{G}_{i,j}$. For
all $(i,j)$ with $c_{i,j} > \epsilon' C$ we hard-wire $x_{i,j}=\mathcal{G}_{i,j}$.

The problem is naturally formulated as the following integer linear program:

\begin{alignat*}{3}
&\sum_{i,j} c_{i,j}x_{i,j} \leq C & ~~~~&(\text{Cost})\\
&\sum_{i} x_{i,j} =y_j  ~~\forall j& ~~~~ &(\text{Assign}) \\
&\sum_{j} p_{i,j}x_{i,j} \leq T ~~\forall i& ~~~~ &(\text{Load})\\
&\sum_{j} \pi_{j}y_{j} \geq \Pi & &(\text{Profit})\\
&x_{i,j} \in \{0,1\}, y_{j} \in \{0,1\} ~~\forall i,j && \\
&x_{i,j}=0 & ~\text{if } p_{i,j} > T & \\
&x_{i,j} = \mathcal{G}_{i,j}&~ \text{ if }
c_{i,j} > \epsilon' C &
\end{alignat*}

We relax the constraint ``$x_{i,j} \in \{0,1\}$ and $y_{j}\in \{0,1\}$'' to
``$x_{i,j} \in [0,1]$ and $y_{j}\in [0,1]$'' to obtain the LP relaxation
\textbf{LP-Out}. We solve the LP to obtain an optimal LP solution $x^{*}, y^{*}$;
we next show how \textbf{Sched-Outlier} rounds $x^{*}, y^{*}$ to obtain
the claimed approximation. 

The rounding proceeds in stages as in Section \ref{sec:sched-cap}.
Each variable maintains its initial assignment in $x^{*}, y^{*}$ in expectation
over the course of rounding. Thus, as we did in Section~\ref{sec:sched-cap}, we drop the cost constraint and finally
derandomize the algorithm to restore this constraint to within a $(1 + \epsilon)$ multiplicative factor as 
claimed by Theorem~\ref{thm:sched-outlier}. 
Also note that if we maintain all the assignment constraints, then the profit constraint can be dropped and
is not violated. Therefore, we consider the profit constraint if and only if 
one or more assignment constraints have been dropped. In addition, we only need to maintain
the total profit obtained from the jobs for which the assignment constraints have been dropped.
 We now proceed  to describe the rounding in each
stage formally.

\medskip

\noindent{\bf Rounding Algorithm.}
Note that $x^{*}_{i,j} \in [0,1]$ denotes the fraction of job $j$ assigned to machine $i$ in $x^{*}$. Initially, $\sum_{i} x^{*}_{i,j}=y^{*}_{j}$.
Initialize $X=x^{*}$.
The algorithm is composed of several iterations;
the random values at the end of iteration $h$ of the overall algorithm are
denoted by $X^{h}$. (Since $y_j=\sum_i x_{i,j}$,
$X^{h}$ is effectively the set of variables.)
Each iteration $h$ (except perhaps the last one)
conducts a randomized update using \textbf{RandMove} on a suitable
polytope constructed from a \textit{subset} of the constraints of \textbf{LP-Out}. Therefore, for all $h$ except perhaps the last, we have
$\expect{X_{i,j}^{h}}=x^{*}_{i,j}$. A variable $X_{i,j}^{h}$ is said to be
\emph{floating} if it lies in $(0,1)$, and a job is \emph{floating} if
it is not yet finally assigned.

\smallskip \noindent \textbf{Key Notation: the current graph $(J, M, E)$.} 
We will throughout take $G = (J, M, E)$ to be the subgraph of the original bipartite graph that is composed of
\textbf{only} the currently-floating edges $(i,j)$. 
We \emph{always} remove degree-$0$ nodes from $G$. 
The following notation always holds: the machines of ``degree'' $k$ in an
iteration are those with exactly $k$ floating jobs assigned fractionally (i.e., those that have degree exactly $k$ in the
current $G$) and similarly, jobs of ``degree'' $k$ are those assigned fractionally to exactly
$k$ machines currently. Note that since we allow $y_j < 1$, there
can exist singleton (i.e., degree-$1$) jobs that are floating. 

Suppose we are at the beginning of some iteration $(h+1)$ of the
overall algorithm; so we are currently looking at the values $X_{i,j}^{h}$.
We will maintain the following invariants:

\smallskip
\noindent 
\textbf{Invariants across iterations:}
\begin{description}
\item[(I1')] Once a variable $x_{i,j}$ gets assigned to $0$ or $1$, it is
never changed;
\item[(I2')] If $j$ is not a singleton, then $\sum_i x_{i,j}$ remains at its
initial value;
\item[(I3')] The constraint (Profit) always holds;
\item[(I4')] Once a constraint is dropped, it is never reinstated. (Recall that even if a constraint is dropped, the variables
associated with it remain.) 
\end{description}

\smallskip

Algorithm \textbf{Sched-Outlier} starts by initializing with $X_{i,j}^{(0)} = x_{i,j}^*$. Iteration $(h+1)$ for $h \geq 0$ consists of four major steps: 

\SN 1. We remove (hard-wire) all $X_{i,j}^{h} \in \{0,1\}$ as in Section \ref{sec:sched-cap}
, i.e., we project $X^{h}$ to those coordinates $(i,j)$ for which
$X_{i,j}^{h} \in (0,1)$, to obtain the current vector $Z$ of ``floating''
variables; let $\mathcal{S} \equiv (A_h Z = u_h)$ denote
the current linear system that represents \textbf{LP-Out}. ($A_h$ is some
matrix and $u_h$ is a vector.)

\SN 2. Let $Z \in \Re^v$ for some $v$; note that $Z \in (0,1)^v$.
Let $M_k$ and $N_k$ denote the set of degree-$k$ machines and degree-$k$
jobs respectively, with $m_k = |M_k|$ and $n_k = |N_k|$.
We will now drop/replace some of the constraints in $\mathcal{S}$:
\begin{description}
\item[(D1')] for each $i \in M_1$, we drop its load constraint from
$\mathcal{S}$;
\item[(D2')] for each $j \in N_1$, we drop its assignment constraint from $\mathcal{S}$. Define a job $j$ to be \emph{tight} if $\sum_{i \in M}Z_{i,j}=1$. By definition $j \in N_1$ are non-tight. Let $J_N$ denote all the non-tight jobs. Hence $N_1 \subseteq J_N$. Maintain a single profit constraint: $$\sum_{j \in {J_N} }Z_{i,j}\pi_{j}=\sum_{j \in J_N} X^{h}_{i,j} \pi_{j}.$$ (Note that at this point, the $X^{h}_{i,j}$
are some known values.)

\end{description}

Thus while the assignment constraints of the singleton jobs are not maintained, their contribution to profit is maintained by
having \emph{one} profit constraint for non-tight jobs. As we noted earlier, it is not required to maintain
the contribution to profit by the non-singleton jobs for
which the assignment constraints are maintained explicitly.

\SN 3. If $Z$ is \textbf{not} a vertex of $\mathcal{S}$, we skip this step and go to Step 4; else if
$Z$ is a vertex of $\mathcal{S}$, we do the following. 
Define the fractional assignment of a machine $i$ by $h_i=\sum_{j \in J} Z_{i,j}$. Drop all the assignment constraints of the non-tight jobs (that is jobs in $J_N$) and maintain a single profit constraint: $$\sum_{j \in N_1 \cup {J_N} }Z_{i,j}\pi_{j}=\sum_{j \in N_1 \cup J_N} X^{h}_{i,j} \pi_{j}.$$ 
While there exists a machine $i'$ whose degree $d$ satisfies $h_{i'} \geq (d-1-\epsilon)$, drop the load constraint on machine $i'$.

\SN 4. Let $\mathcal{P}$ denote the polytope defined by the current system of
constraints. If $Z$ is not a vertex of $\mathcal{P}$, invoke \textbf{RandMove}($Z,\mathcal{P}$).
(\textbf{Comment:} As usual, invoking \textbf{RandMove} leads to progress for us, since it reduces the number of floating variables by at least $1$, or increases the number of tight constraints by at least $1$.) Else 
(in this case the current iteration will be the \textbf{last} iteration) we proceed as follows
 depending on the configuration of machines and jobs in the system and \textbf{halt}. If none of the following configurations is achieved (which we will
 show never happens at a vertex), then we report error and exit. There are five possible configurations, which we
describe next along with the steps we take for each.

\SN {\bf Config-1:} \emph{The Machine-job bipartite graph $(J,M,E)$ consists only of vertex-disjoint cycles.} In this
configuration, we orient the edges in the bipartite graph to assign the jobs in $(J,M,E)$ in such a way that each machine gets at most one job.  Note that such an orientation is easy with disjoint cycles since they have even lengths.

\SN {\bf Config-2:} \emph{The Machine-job bipartite graph $(J,M,E)$ consists of vertex-disjoint cycles and exactly one path -- that is vertex-disjoint from the cycles -- that has both end-points being job nodes. Thus there are two singleton jobs.} In this
case, we discard one among the two singleton jobs that has less profit; we again orient the edges in the remaining bipartite graph to assign the remaining jobs such that each machine gets at most one job. 

\SN {\bf Config-3:} \emph{There is exactly one job of degree $3$ and one singleton job; the rest of the jobs have degree $2$ and all of the machines have degree $2$.} 
 Here we assign the singleton job to the degree-$2$ machine it is fractionally attached to and remove the other edge (but not the job) associated with that machine.
  We are left with disjoint cycles. Orient the edges in the cycles of the bipartite graph to assign the remaining jobs in such a way that each machine gets at most one job.

\SN {\bf Config-4:} \emph{There is only one degree-$3$ machine with {\it one} singleton job attached to it; the rest of the machines have exactly
 two non-singleton jobs attached to each of them fractionally. Each non-singleton job is attached fractionally to exactly two machines.}
In this configuration, we assign the singleton job and the cheaper (less processing time) of the two non-singleton jobs to the degree-$3$ machine. The rest of the
  jobs and the machines form disjoint cycles in the machine-job bipartite graph or form disjoint paths, each with the number of machines in it and the number of jobs in it being the same. 
  Orient the edges in this remaining bipartite graph in such a way that each machine gets one among the two jobs fractionally attached to it.

\SN {\bf Config-5}: \emph{The machine-job bipartite graph $(J,M,E)$ consists of vertex-disjoint cycles and exactly one extra edge with one singleton job and one singleton machine.} Here, we 
 assign the singleton job to the singleton machine. Orient the edges in the cycles of the bipartite graph to assign the remaining jobs in such a way that each machine gets at most one job.

\begin{figure}
\begin{centering}
    \includegraphics[scale=0.4]{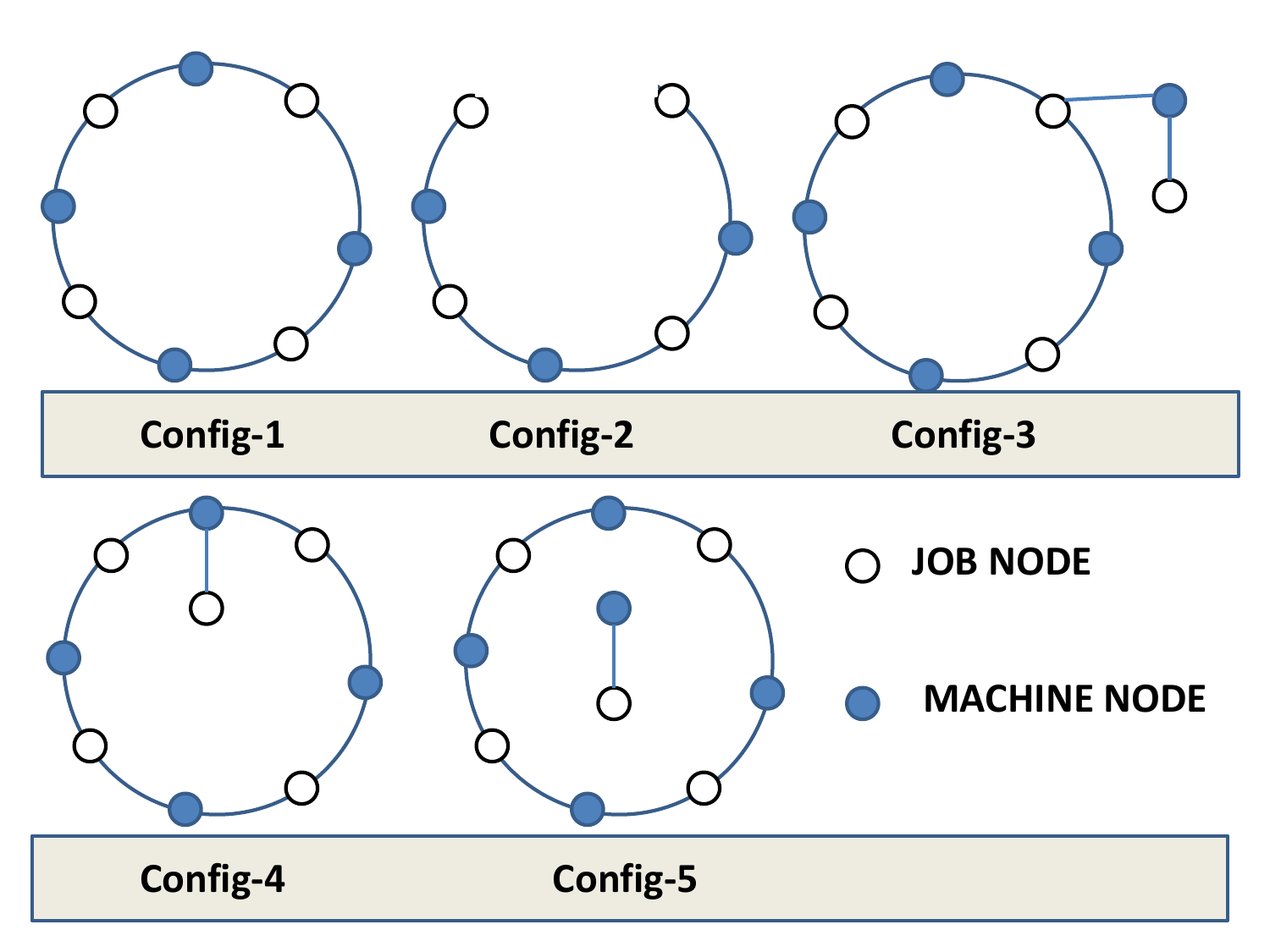}
    \caption{Different configurations of the machine-job bipartite graph in step 4 of Sched-Outlier}
     \end{centering}
 \label{fig:config}
  \end{figure}

The different configurations are shown pictorially in the figure. This ends the description of the algorithm.

 \medskip

\paragraph{{\bf Analysis}}
Our analysis follows the following structure. First, we prove two key lemmas -- Lemma \ref{lem:key1} and Lemma \ref{lem:key2} -- which show that if $Z$ is a vertex and the algorithm reaches step $4$, then
one of the five configurations as described above happens and also that the number of machines in $G$ is 
lesser than
$\frac{1}{\epsilon}$. Lemmas \ref{lem:key1}  and \ref{lem:key2} are followed by Lemma \ref{lem:profit}. Lemma \ref{lem:profit}
establishes that the dropping and the modification of constraints
in step 2 and 3, along with the assignment of jobs in step 4 do not violate the load constraint by more than a factor of $(2+\epsilon)$
and maintain the profit constraint.  Lemma \ref{lem:cost} bounds the cost.

Recall that in the bipartite
graph $G=(J,M,E)$, we have in iteration $(h+1)$ that $(i,j) \in E$
iff $X_{i,j}^{h} \in (0,1)$; also, any job or machine having degree $0$ is \textbf{not} part of $G$.

\begin{lemma}
\label{lem:key1}
If $Z$ is a vertex of the polytope at the beginning of step $3$, then the following must be true at the beginning of step $3$: (i) one of the five configurations described in step $4$ must occur then, and (ii) the number of floating variables must equal the number of constraints in our system then. 
\end{lemma}


\begin{proof}
Let us consider the different possible configurations of $G$ when $Z$ becomes a vertex of the polytope $\mathcal{P}$ at the beginning of step $3$.  
There are several cases to consider depending on the number of singleton floating jobs in $G$ in that iteration. For each
case, we will prove (i) and (ii). 

\SN \textbf{Case 1:} There is no singleton job. 
We have $n_1=0$. Then, the number of constraints in $\mathcal{S}$ is $$EQ=\sum_{k \geq 2} m_{k} + \sum_{k \geq 2} n_{k}.$$ Recall that since there is no
singleton job, we do not consider the profit constraint explicitly. The number of floating variables is $v=\sum_{k \geq 2} k n_{k}$; alternatively, $v=\sum_{k \geq 1} k m_{k}$. Therefore, $$v=\sum_{k \geq 2} \frac{k}{2} (m_{k}+n_{k}) +\frac{m_1}{2}.$$ $Z$ being a vertex of $\mathcal{P}$, $v \leq EQ$. Thus, we must have $n_k = m_k=0$ for all $k \geq 3$ and $m_1=0$. Hence, every floating machine has exactly two floating jobs assigned to it and every floating job is assigned exactly to two floating machines. This is handled by Config-1, which also satisfies $v = EQ$ as required by part (ii) of the
lemma.  

\SN \textbf{Case 2:} There are at least three singleton jobs.
We have $n_1 \geq 3$. Then the number of linear constraints is $EQ=\sum_{k \geq 2} m_{k} + \sum_{k \geq 2} n_{k} + 1$, where the last ``$1$'' comes from the single profit constraint. The number of floating variables $v$ again by the averaging argument as above is
 $$v=\frac{n_1}{2}+\sum_{k \geq 2} \frac{k}{2} (m_k+n_{k}) +\frac{m_1}{2} \geq \frac{3}{2} + \sum_{k \geq 2} \frac{k}{2}( m_{k}+n_{k}) +\frac{m_1}{2}.$$ Hence, the system is always underdetermined and $Z$ cannot be a vertex of $\mathcal{P}$.

\SN \textbf{Case 3:} There are exactly two singleton jobs.
We have $n_1=2$. Then the number of linear constraints is $$EQ=\sum_{k \geq 2} m_{k} + \sum_{k \geq 2} n_{k} + 1;$$ again the last ``$1$'' comes from the single profit constraint. The number of floating variables $v$ by the averaging argument is
 $$v=\frac{n_1}{2}+\sum_{k \geq 2} \frac{k}{2} (m_k+n_{k}) +\frac{m_1}{2} \geq 1 + \sum_{k \geq 2} \frac{k}{2}( m_{k}+n_{k}) +\frac{m_1}{2}.$$ Thus we must have $n_k = m_k = 0$ for all $k \geq 3$, and $m_1=0$ -- and
thus also that $v = EQ$. Also, every floating machine has exactly two floating jobs assigned to it and each job -- except for the two singleton jobs -- is assigned to exactly two machines fractionally: this is handled by Config-2.

\SN \textbf{Case 4:} There is exactly one singleton job.
We have $n_1=1$ here. Then the number of linear constraints is 
\begin{equation}
\label{case4:EQ}
EQ=\sum_{k \geq 2} m_k + \sum_{k \geq 2} n_k + 1.
\end{equation}
The number of floating variables is
\begin{equation}
\label{case4:v}
v \geq \frac{1}{2}+n_{2}+\frac{3}{2}n_{3}+\frac{m_1}{2}+m_2+\frac{3}{2}m_{3}+\sum_{k \geq 4} \frac{k}{2}(m_k+n_{k}).
\end{equation}
 If $Z$ is a vertex of $\mathcal{P}$, then $v \leq EQ$. There are only three possible sub-cases that might arise in this case: 

\SN (i) $n_3 = 3$ (and $n_1 = 1$). It is easy to check here that for the r.h.s.\ of (\ref{case4:v}) to be upper-bounded by
the r.h.s.\ of (\ref{case4:EQ}), all the other jobs must have degree $2$ and all the machines must have degree $2$. This is handled by Config-3, and we have $v = EQ$. 

\SN (ii) $m_3 =1$ (and $n_1 = 1$). Just as in sub-case (i), the rest of the jobs  and machines must have degree $2$. This is handled by Config-4; we again have $v = EQ$ here.

\SN (iii) $m_1 = n_1 = 1$. The rest of the jobs and machines have degree $2$. This is handled by Config-5 and again 
satisfies $v = EQ$. 
\end{proof}

 \begin{lemma}
\label{lem:key2}
(a) Let $m$ denote the number of machine-nodes in $G=(J,M,E)$ at the beginning of step $4$. 
If $m \geq \frac{1}{\epsilon}$, then $Z$ is not a vertex of the polytope at the beginning of step $4$.
(b) If $Z$ was a vertex of the polytope at the beginning of step $4$ and if Config-2 held during this step, 
then the total fractional assignment of the two singleton jobs is less than $1$. 
\end{lemma}
\begin{proof}
Most of this proof is centered on (a); we handle (b) when we address Config-2 below. 

Suppose $Z$ is a vertex of the polytope at the beginning of step $4$. Then by Lemma~\ref{lem:key1}(i), one of the five configurations described in step $4$ must occur. Our strategy now is to show that if
$m \geq \frac{1}{\epsilon}$, then it \emph{cannot be} that one of the following two happened in step $3$:
(a) {\it we were not able to drop any constraint in step $3$}, or (b) {\it we dropped exactly one constraint -- which was 
an assignment constraint for a non-tight job -- in step $3$ but also added one profit constraint in step $3$}. 
Given  Lemma~\ref{lem:key1}(ii), the impossibility of (a) and (b) would then show that our system is underdetermined at
the beginning of step $4$ if $m \geq \frac{1}{\epsilon}$ as required. 

In any configuration, if there is a cycle with all tight jobs, then there always exists a machine with total fractional assignment at least $1$ and hence its load constraint is dropped in step $3$ -- as its degree is $2$. So we assume there is \emph{no such cycle} in any configuration, since the proof is complete otherwise. 

Now suppose the algorithm reaches Config-1. If there are two non-tight jobs, then we drop two assignment constraints and only add one profit constraint. Thus the system becomes underdetermined. Therefore, there can be at most one non-tight job and only one cycle overall (say $C$), since we have assumed above that there is no cycle with all jobs tight. Let $C$ have $m$ machines and thus $m$ jobs. Therefore, $\sum_{i,j \in C} x_{i,j} \geq m-1$. Thus there exists a machine such that the total fractional assignment of jobs on that machine is at least $\frac{m-1}{m}=1-1/m$. If $m \geq \frac{1}{\epsilon}$, then there exists a machine  with degree $2$ and with total fractional assignment at least $(1-\epsilon)$: thus the load-constraint on that machine gets dropped, making the system underdetermined.

Suppose the algorithm reaches Config-2: we also handle part (b) of the Lemma here.
In this case, all the non-singleton jobs must be tight for $Z$ to be a vertex. If there are $m$ machines, then the number of non-singleton jobs is $m-1$. Let the two singleton jobs be $j_1$ and $j_2$, and the two machines to which jobs $j_1$ and $j_2$ are fractionally attached with be $i_1$ and $i_2$ respectively. If $x_{i_1,j_1}+x_{i_2,j_2} \geq 1$, then the total fractional assignment from all the jobs in the system is at least $m$; thus the machine with maximum
fractional assignment must have an assignment at least 1. Since this machine has degree $2$, its load constraint would have been dropped -- a contradiction, thus also proving (b). Thus, the only case to consider (for part (a)) is that $x_{i_1,j_1}+x_{i_2,j_2} < 1$, where the total fractional assignment of all the jobs in the system is at least $m-1$. Thus there exists a machine such that the total fractional assignment of jobs on that machine is $\geq \frac{m-1}{m}=1-1/m$. If $m \geq \frac{1}{\epsilon}$, then there exists a machine  with degree $2$ and with total fractional assignment 
at least $(1-\epsilon)$. Hence the load constraint on that machine gets dropped, making the system underdetermined.

 For Config-3 and Config-5, if $Z$ is a vertex of $\mathcal{P}$, then all the non-singleton jobs must be tight and using essentially the same argument as above, there exists a machine with fractional assignment at least $(1-\epsilon)$ if the algorithm reaches Config-3 and there exists a machine with fractional assignment $1 > 2 - 1 - \epsilon$ if the algorithm reaches Config-5.

If the algorithm reaches Config-4, then again all the non-singleton jobs must be tight. If the degree-$3$ machine has fractional assignment at least $2-\epsilon$, then its load constraint can be dropped to make the system underdetermined. Otherwise, the total assignment to the degree-$2$ machines from all the jobs in the cycle is at least $m-2+\epsilon$. Therefore, there exists at least one degree-$2$ machine with fractional assignment at least $\frac{m-2+\epsilon}{m-1}=1-\frac{1-\epsilon}{m-1}\geq 1-\epsilon$, if $m \geq \frac{1}{\epsilon}$. The load-constraint on that machine will be dropped in step $3$. 

Hence, it is not possible that $Z$ is a vertex of the polytope in step $4$ if the number of machines $m$ is at least $\frac{1}{\epsilon}$. 
This completes the proof of Lemma \ref{lem:key2}.
\end{proof}

\medskip 
We next show that with probability $1$, the final profit is at least $\Pi$ and the final makespan is at most $(2+\epsilon)T$:

 \begin{lemma}
 \label{lem:profit}
 Let $X$ denote the final rounded vector. Algorithm \textbf{Sched-Outlier} returns a schedule, where with probability one, (i) the profit is at least $\Pi$, (ii) for all $i$, $\sum_{j \in J}X_{i,j}p_{i,j} < \sum_{j}x^{*}_{i,j}p_{i,j}+(1+\epsilon)\text{max}_{j \in J:~x^{*}_{i,j}\in (0,1] }p_{i,j}$.
 \end{lemma}
\begin{proof}
(i) This essentially follows from the fact that whenever the assignment constraint for any job is dropped, its profit constraint is included in the global profit constraint of the system. In step $4$, with the exception of one configuration (Config-2), all the jobs are always assigned; thus the profit cannot decrease in these other configurations. In Config-2, since we are at a vertex
in step $4$, Lemma \ref{lem:key2}(b) shows that the total fractional assignment of the two singleton jobs is less than $1$.  Thus a singleton job (say $j_1$) is dropped only when $G$ has two singleton jobs $j_1,j_2$ fractionally assigned to $i_1$ and $i_2$ respectively, with total assignment $x_{i_1,j_1}+x_{i_2,j_2} < 1$.  Since the job with the higher profit is retained, $\pi_{j_1}x_{i_1,j_1}+\pi_{j_2}x_{i_2,j_2} \leq max\{\pi_{j_1},\pi_{j_2}\}$.

(ii) A machine's fractional load is preserved until its load constraint is dropped (if at all). When can such a load constraint be
dropped in an iteration? Note from \textbf{(D1')} that load constraints are dropped from machines $i \in M_1$; Lemma \ref{lem:key1} implies that the load constraint also might be dropped from some machine(s) $i \in M_2 \cup M_3$ in
step $3$. For $i \in M_1$, only the remaining job $j$ with $X^{h}_{i,j}>0$ can get fully assigned to it any further. Hence for $i \in M_1$, its total load is less than $\sum_{j}x^{*}_{i,j}p_{i,j}+\text{max}_{j \in J:x^{*}_{i,j}\in (0,1]}p_{i,j}$. For any machine $i \in M_2 \cup M_3$, if its degree $d$ ($2$ or $3$) is such that its fractional assignment is at least $d-1-\epsilon$, then by simple algebra, it can be shown that for any such machine $i$,  its total load is at most $\sum_{j}x^{*}_{i,j}p_{i,j}+(1+\epsilon)\text{max}_{j \in J:x^{*}_{i,j}\in (0,1]}p_{i,j}$ at the end of the algorithm. For the remaining machines consider what happens in step 4. Since this is the last iteration, it suffices to show that the load does not increase by too much in this last iteration. Except when Config-4 is reached, any remaining machine $i$ gets at most one extra job, and thus its total load is less than $\sum_{j}x^{*}_{i,j}p_{i,j}+\text{max}_{j \in J:~x^{*}_{i,j}\in (0,1]}p_{i,j}$. When Config-4 is reached in step 4, if the degree-$3$ machine (say $i$) has a fractional assignment some $f \leq 1$ from the two jobs in the cycle, then 
the total fractional load on the remaining $m - 1$ machines is $m -f \geq m -1$, which means that one of these 
(degree-$2$) machines had a load of at least $1$: this in turn means that such a machine would have had its load
constraint dropped in step $3$, which is a contradiction. Hence, let $j_1,j_2,j_3$ be the three jobs assigned fractionally to machine $i$ and let $j_3$ be the singleton job; as argued in the previous sentence, $x_{i,j_1}+x_{i,j_2} > 1$. If  $p_{i,j_1} \leq p_{i,j_2}$, then machine $i$ gets jobs $j_1$ and $j_3$ assigned to it; else $i$ gets $j_2, j_3$. Since the fractional assignment on $i$ from $j_1$ and $j_2$ is more than $1$ and since the job with less processing time among $j_1$ and $j_2$ is assigned to $i$, $i$'s final load is less than $\sum_{j}x^{*}_{i,j}p_{i,j}+\text{max}_{j \in J:~x^{*}_{i,j}\in (0,1]}p_{i,j}$.
This completes the proof of Lemma \ref{lem:profit}. 
\end{proof}

Finally we have the following lemma.

 \begin{lemma}
 \label{lem:cost}
Algorithm \textbf{Sched-Outlier} can be derandomized to output a schedule of cost at most $C(1+\epsilon)$.
 \end{lemma}

\begin{proof}
In all iterations $h$ except possibly the last one, we have for all $i, j$ that $E[X^{h}_{i,j}]=x^{*}_{i,j}$, where $x^{*}_{i,j}$ is solution of \textbf{LP-Out}. Therefore, before the last iteration, we have that the total expected cost
incurred is $C$. As in Section~\ref{sec:sched-cap}, the procedure can be derandomized directly by the method of
conditional expectations, giving an $1$-approximation to cost, just before the last iteration. Now in the last iteration, since at most $\frac{1}{\epsilon}$ jobs are assigned and each assignment requires at most $\epsilon'C=\epsilon^2 C$ in cost, the total increase in cost is at most $\epsilon C$, giving the required approximation.
\end{proof}

Lemmas \ref{lem:profit} and \ref{lem:cost} yield Theorem \ref{thm:sched-outlier}.

\smallskip
We next consider Theorem \ref{thm:sched-fairness} that maintains fairness in the allocation of jobs while handling outliers.

{\bf Proof of Theorem \ref{thm:sched-fairness}:}
\begin{proof}
We consider the LP-relaxation \textbf{LP-Out} except that: (i) we add the constraints $y_j = r_j$, and (ii) drop the constraint $x_{i,j} = \mathcal{G}_{i,j}~ \text{ if }
c_{i,j} > \epsilon' C$ (i.e., in order to maintain the scheduling probabilities of the jobs, we
do not guess the assignment of jobs with high cost). 

For part (i), we consider the first two steps of Algorithm {\bf Sched-Outlier}. If
$\mathcal{P}$ denotes the polytope defined by the reduced system of constraints and
the current vector $Z$ is not a vertex of $\mathcal{P}$, then we invoke {\bf RandMove}$(Z,P)$
and proceed. Else from Lemma \ref{lem:key1}, $Z$ is a vertex of $\mathcal{P}$ only if one of the
configurations, Config-1 to Config-5, as described in step 4 of Algorithm {\bf Sched-Outlier} is achieved
and $m < \frac{1}{\epsilon}$.
For any singleton job, we assign the
singleton job to the corresponding machine with probability equal to its fractional assignment. Thus
Theorem \ref{thm:sched-fairness} remains valid for these singleton jobs. For each non-singleton job,
we consider the machines to which it is fractionally assigned and allocate it to the machine
which has cheaper assignment cost for it. If the algorithm reached Config-1, 2, 3 or 5, each machine
 can get at most two extra jobs and the expected cost is maintained. However if
 the algorithm reached Config-4 and the three jobs associated with the degree-3
 machine were all assigned to it, then we remove one non-singleton job from the
 degree-3 machine. This job is assigned to the degree-2 machine in the cycle on
 which it had non-zero fractional assignment. This may increase the expected cost
 by a factor of $2$ but ensures that each machine gets at most $2$ additional jobs.

\end{proof}

\section{Max-Min Fair Allocation}
\label{sec:max-min}
In this section we consider another application that has received significant attention in the recent past: the max-min fair allocation problem
\cite{dani:05,julia:focs09,asadpour:stoc07,asadpour-feige-saberi,bansal:stoc06}. 
We provide a new algorithm for max-min fair allocation based on bipartite dependent rounding \cite{gkps:dep-round} and its generalization to weighted graphs. Bipartite dependent rounding has found many applications in combinatorial optimization \cite{srin:level-sets,gkps:dep-round,kmps:unified-sched}, and can be seen as a special case of {\bf RandMove} on bipartite graphs. We also consider an ``equitable allocations'' version of
such problems, in Theorem~\ref{thm:santa2}: this theorem follows from \cite{shmoys-tardos:gap} as pointed out by the referee. 

In the max-min fair allocation problem,
there are $m$ goods that need to be distributed indivisibly among $k$ persons. Each person $i$ has a non-negative integer valuation $u_{i,j}$ for good $j$. The valuation functions are linear, i.e., $u_{i,C}=\sum_{j \in C} u_{i,j}$ for any set of $C$ goods. The goal is to allocate each good to a person such that the ``least happy person is as happy as possible'': i.e., $\min_{i} u_{i,C}$ is maximized. Our main contribution  in this regard
is to near-optimally pin-point the integrality gap of a {\em configuration LP} previously proposed and analyzed in \cite{bansal:stoc06,asadpour:stoc07}.

\noindent
\paragraph{\textbf{The Configuration LP for Max-Min Fair Allocation}}

The configuration LP formulation for the max-min fair allocation problem was first considered in
\cite{bansal:stoc06}. A {\it configuration} is a subset of items,
and the LP has a
variable for each valid configuration. Using binary search, first the optimum solution value $T$ is guessed and then
we define valid configurations based on the approximation factor $\lambda$ sought; we will set
\begin{equation}
\label{eqn:lambda}
\lambda = 26 \sqrt{k \ln k}.
\end{equation}
We call a configuration $C$  {\it valid for person $i$} if either of the following two conditions hold:
\begin{itemize}
\item  $u_{i,C} \geq T$ and each item  in $C$ has value less than $\frac{T}{\lambda} $. These are called {\em small} items.
\item $C$ contains only one item $j$ and $u_{i,j} \geq \frac{T}{\lambda}$. We call such an item $j$ to be a {\em big} item for person $i$.
\end{itemize}

We define a variable $x_{i,C}$ for assigning a valid configuration $C$ to person $i$.
Let $C(i,T)$ denote the set of all valid configurations corresponding to person $i$ with respect to $T$. The configuration LP relaxation
 of the problem is as follows:

\begin{eqnarray}
\label{eqn:lp}
\forall j: \sum_{C \ni j}\sum_{i} x_{i,C} \leq 1 \\
\forall i: \sum_{C \in C(i,T)} x_{i,C}=1 \nonumber\\
\forall i, C: x_{i,C} \geq 0 \nonumber
\end{eqnarray}

The above LP formulation may have an exponential number of variables, However, if the LP is feasible, then a fractional allocation where each person receives either a big item or at least a utility of $T(1-\epsilon)$ can be computed in
polynomial time for any constant $\epsilon > 0$ \cite{bansal:stoc06}.
In the subsequent discussion and analysis, we ignore the multiplicative
$(1 - \epsilon)^{-1}$ factor; it is hidden in the $\Theta$ notation of
the ultimate approximation ratio.

The worst-case integrality gap of the above configuration LP is
lower-bounded by $\Omega(\frac{1}{\sqrt{k}})$ \cite{bansal:stoc06}.
In \cite{asadpour:stoc07}, Asadpour and Saberi
gave a rounding procedure for the configuration LP that achieved an approximation factor of $O\left(\frac{1}{\sqrt{k}(\ln{k})^3}\right)$. Here we further lower
the gap and prove the following theorem;  our proof is also
significantly simpler than that of \cite{asadpour:stoc07}.

\begin{theorem}
\label{thm:santa1}
Given any feasible solution to the configuration LP, it can be rounded to a feasible integer solution such that every person gets  at least $\Theta\left(\frac{1}{\sqrt{k\ln{k}}}\right)$ fraction of the optimal utility with probability at least $1-\Theta(\frac{1}{k})$, in polynomial time.
\end{theorem}

Note that the work of Chakrabarty, Chuzhoy and Khanna \cite{Chuzhoy09}
yields an improved approximation factor of $m^{\epsilon}$
for any positive constant
$\epsilon$, but it does not use the configuration LP
(also note that $m \geq k$).

In the context of fair allocation, an additional important criterion can be
an \emph{equitable partitioning} of goods: we may impose an upper
bound on the number of items a person might receive. For example,
we may want each person to receive at most $\lceil \frac{m}{k} \rceil$ goods. Theorem \ref{thm:matching} then directly leads to the following.
\begin{theorem}
\label{thm:santa2}
Suppose, in max-min allocation, we are given upper bounds $c_{i}$ on the
number of items that each person $i$ can receive, in addition to the
 utility values $u_{i,j}$. Let $T$ be the optimal max-min allocation value
that satisfies $c_i$ for all $i$. Then, we can efficiently construct an allocation
in which for each person $i$ the bound $c_i$ holds and she receives a total
utility of at least $T - \max_{j} u_{i,j}$.
\end{theorem}

This generalizes the result of \cite{dani:05}, which yields the
``$T - \max_{j} u_{i,j}$'' value when no bounds such as the $c_i$ are given.
To our knowledge, the results of
\cite{julia:focs09,asadpour:stoc07,asadpour-feige-saberi,bansal:stoc06}
do not carry over to the setting of such ``fairness bounds'' $c_i$.

\subsection{Algorithm for Max-Min Fair Allocation}
\label{subsec:algo}

We now describe the algorithm and proof for Theorem \ref{thm:santa1}.

\subsubsection{Algorithm}

We define a weighted bipartite graph $G$ with the vertex set $A \bigcup B$ corresponding to the persons and the items respectively. There is an edge between a vertex corresponding to person $i \in A$ and item $j \in B$, if a configuration $C$ containing $j$ is fractionally assigned to $i$. Define
$$w_{i,j}=\sum_{C \ni j} x_{i,C},$$
i.e., $w_{i,j}$ is the
fraction of item $j$ that is allocated to person $i$ by the fractional solution of the LP. An edge $(i,j)$ is called a {\em matching edge}, if the item $j$ is big for person $i$. Otherwise it is called a {\em flow edge}.

Let $M$ and $F$ represent the set of matching and flow edges respectively. For each vertex $v \in A \bigcup B$, let $m_{v}$ denote the total fractional weight of the matching edges incident to it. That is if $v$ is a person then $$m_v=\sum_{j \text{ is a big item for } v} w_{v,j}=\sum_{j \in C, C \text{ contains the big item $j$ for } v} x(v,C).$$ And if $v$ is a job then $$m_v=\sum_{i: \text{ $v$ is a big item for } i} w_{i,v}=\sum_{i,  C \text{ contains $v$ as a big item for } i} x(i,C).$$  Also define $f_{v}=1-m_v$. The main steps of the algorithm are as follows.

\bigskip

\begin{enumerate}
\item[1.] Guess the value of the optimal solution $T$ by doing a binary search. Solve LP (\ref{eqn:lp}).
  Obtain the set $M$ and $m_v, f_v$ for each vertex $v$ in $G$ constructed from the LP solution.

  \bigskip

\item[2] {\bf Allocating Big Items}: Select a random matching from edges in $M$ using { \em bipartite dependent rounding}
(see Section \ref{subsec:match}) such that for every $v \in A \bigcup B$, the probability that $v$ is matched by the matching is $m_v=1-f_v$.

\bigskip

\item[3] {\bf Allocating Small Items}: Let $\epsilon_1=\sqrt{\frac{\ln{k}}{k}}$.
\begin{enumerate}
\item Discard any item $j$ with $m_{j} \geq (1-\epsilon_1)$, and also discard all the persons and the items matched by the matching.
\item (Scaling) In the remaining graph containing only flow edges for unmatched persons and items, set for each
person $i$, $w'_{i,j}=\frac{w_{i,j}}{f_{i}}, ~ \forall j$.
\item Further discard any item  $j$ with $\sum_{i} w'_{i,j} \geq \psi(k)$, where $\psi(k)$ is defined below.
\item Scale down the weights on all the remaining edges by a factor of $\psi(k)$ and run the algorithm of \cite{dani:05} to assign the small items.
\end{enumerate}

matched
\end{enumerate}

 \noindent{\bf Choice of $\psi(k)$.}

Let us consider the functions $\Phi(k) = 100 \ln\ln\ln k / \ln\ln k$ (note that
$\Phi$ is asymptotically zero.) and
\begin{equation}
\label{eqn:psi-bound}
\psi = \psi(k) = \frac{3 \ln k}{\ln \ln k} \cdot (1 + \Phi(k));
\end{equation}
For large enough $k$, say $k \geq 10$, the following holds.
\begin{equation}
\label{eqn:Phi-bound}
(1 + \psi) \ln(1 + \psi) - \psi \geq 3 \ln k
\end{equation}
 This is easily verified by plugging the fact $\ln(1 + \psi) \geq \ln\ln k - \ln\ln\ln k$ into (\ref{eqn:Phi-bound}).

\bigskip

We now analyze each step. The main proof idea is in showing that there remains enough left-over utility in the flow graph for each person
not matched by the matching. This is obtained through proving a negative correlation property among the random variables
defined on a collection of vertices. Previously, the negative correlation property due to bipartite dependent
rounding was known for variables defined on edges incident on any particular vertex. We adapt the proof according to our need.

\subsubsection{Allocating Big Items}
\label{subsec:match}
Consider the edges in $M$ in the person-item bipartite graph. Remove all the edges $(i,j)$ that have already been rounded to $0$ or $1$. Additionally, if an edge is rounded to $1$, remove both its endpoints $i$ and $j$.  We initialize for each $(i,j) \in M$, $y_{i,j}=w_{i,j}$, and modify the $y_{i,j}$ values probabilistically in rounds using bipartite dependent rounding.

\paragraph{{\bf Bipartite Dependent Rounding}\cite{gkps:dep-round}}
\label{subsec:bipartite}
 We give a brief sketch of bipartite dependent rounding introduced in \cite{gkps:dep-round} for the sake of completeness.

  The bipartite dependent rounding selects an even cycle $\mathcal{C}$ or a maximal path $\mathcal{P}$ in $G$, and partitions the edges in $\mathcal{C}$ or $\mathcal{P}$ into two matchings $\mathcal{M}_{1}$ and $\mathcal{M}_{2}$. Then, two positive scalars $\alpha$ and $\beta$ are chosen as follows:
 $$\alpha =\min\{\eta > 0: ((\exists(i,j) \in \mathcal{M}_{1}: y_{i,j}+\eta=1) \bigcup(\exists(i,j) \in \mathcal{M}_{2}: y_{i,j}-\eta=0))\};$$
 $$\beta =\min\{\eta > 0: ((\exists(i,j) \in \mathcal{M}_{1}: y_{i,j}-\eta=0)\bigcup(\exists(i,j) \in \mathcal{M}_{2}: y_{i,j}+\eta=1))\};$$

Now with probability $\frac{\beta}{\alpha+\beta}$, set
\begin{eqnarray*}
&& y'_{i,j}=y_{i,j}+\alpha \text{ for all } (i,j) \in \mathcal{M}_{1} \\
\text{ and } && y'_{i,j}=y_{i,j}-\alpha \text{ for all } (i,j) \in \mathcal{M}_{2};
\end{eqnarray*}
with complementary probability of $\frac{\alpha}{\alpha+\beta}$, set
\begin{eqnarray*}
&& y'_{i,j}=y_{i,j}-\beta \text{ for all } (i,j) \in \mathcal{M}_{1} \\
\text{ and }&& y'_{i,j}=y_{i,j}+\beta \text{ for all } (i,j) \in \mathcal{M}_{2};
\end{eqnarray*}

The above rounding scheme satisfies the following two properties, which are easy to verify:
\begin{equation}
\label{dep:prop1}
\forall \,i,j,\, \expect{y'_{i,j}}=y_{i,j}
\end{equation}
\begin{equation}
\label{dep:prop2}
\exists\,i,j,\, y'_{i,j}\in \{0,1\}
\end{equation}

Thus, if $Y_{i,j}$ denotes the final rounded values then Property (\ref{dep:prop1}) guarantees for every edge $(i,j)$,  $\expect{Y_{i,j}}=w_{i,j}$. This gives the following corollary.

\begin{corollary}
\label{cor:match1}
The probability that a vertex $v \in A \bigcup B$ is matched in the matching generated by the algorithm is $m_v$.
\end{corollary}

\begin{proof}
Let there be $l \geq 0$ edges $e_1,e_2,..e_l \in M$ that are incident on $v$. Then,
\begin{eqnarray*}
\prob{v~ \text{is matched}}&=&\prob{\exists\, e_i, i \in [1,l]~~ s.t~~ v~\text{is matched with }~e_i}\\
&=&\sum_{i=1}^{l} \prob{ v \text{ is matched with } e_{i}}=\sum_{i=1}^{l} w_{i}=m_{v}
\end{eqnarray*}

Here the second equality follows by replacing the union bound by sum since the events are mutually exclusive.
\end{proof}

\paragraph*{Negative Correlation over Multiple Vertices} 
Now we show additional properties of this rounding to be used crucially for the analysis of the next step.
Recall the notion of negative correlation from Definition~\ref{defn:neg-correl}.
We show a useful negative-correlation property for dependent rounding on bipartite graphs over multiple vertices. The proof is syntactically similar to Lemma~2.2 of \cite{gkps:dep-round}. However, \cite{gkps:dep-round} only shows negative correlation property for random variables defined on edges incident to a single vertex; here a stronger negative correlation property is proven for random variables defined on multiple vertices. We state the theorem here, and prove it in the appendix. 


\begin{theorem}
\label{theorem:neg1}
Define an indicator random variable $z_{j}$ for each item $j \in B$ with $m_{j} < 1$, such that $z_{j}=1$
 if item $j$ is matched by the matching. Then, the indicator random variables $\{z_{j}\}$ are negatively correlated.
 \end{theorem}

 As a corollary of Theorem \ref{theorem:neg1}, we get the following:
 \begin{corollary}
\label{cor:neg2}
Define an indicator random variable $z_{i}$ for each person $i \in A$, such that $z_{i}=1$
 if person $i$ is matched by the matching. Then, the indicator random variables $\{z_{i}\}$ are negatively correlated.
 \end{corollary}

\begin{proof}
Do the same analysis as in Theorem \ref{theorem:neg1} with items replaced by persons.
\end{proof}

\subsubsection{Allocating small items}
\label{subsec:alloc}

We start by proving in Lemma \ref{lemma:small} that after the matching phase, we have with high probability that
each unmatched person has available items with utility at least $\sqrt{\frac{\ln{k}}{k}}\frac{T}{5}$ in the flow graph. Additionally we prove in Lemma  \ref{lemma:item} that given any \emph{particular} item $j$, we have with probability at least $1 - O(1/k)$ that $j$ is claimed at most
$\psi(k)$ times. Note that this probability is not large enough to afford a union bound over all the $m$ possible values of $j$, since $m$ is not
bounded as a function of $k$; Lemma~\ref{lemma:stepc} shows how to get around this issue. Both of these probabilistic results use
Theorem~\ref{thm:chbound-negcorrel}.

%

\begin{lemma}
\label{lemma:small}
After Step 2 of  allocation of big items by bipartite dependent rounding, we have the probability for all unmatched person to have a total utility of at least
$\sqrt{\frac{\ln{k}}{k}}\frac{T}{5}$ from the unmatched items is at least $1-\frac{1}{k}$.
\end{lemma}

\begin{proof}
Consider a person $v$ who is unsatisfied by the matching. Define $w'_{v,j}=\frac{w_{v,j}}{f_{v}}$. Then according to LP (\ref{eqn:lp}) solution
\begin{equation}
\sum_{j}w'_{v,j} u_{v,j} =T
\end{equation}

In step (a) of {\bf Allocation of Small Items}, all items $j$ with $m_{j}$
 at least $(1-\epsilon_1)$ are discarded; recall that $\epsilon_1=\sqrt{\frac{\ln{k}}{k}}$. Since the total sum of $m_j$ can be at most $k$ (the number of persons), there can be at most $\frac{k}{1-\epsilon_1}$ items with $m_j$ at least $1-\epsilon_1$. Therefore, for the remaining items, we have $f_j \geq \epsilon_1$. Each person is connected only to small items in the flow graph. After removing the items with $m_j$ at least $1-\epsilon_1$, the remaining utility in the flow graph for person $v$ is  at least
 \begin{equation}
 \label{eq:1}
 \sum_{j: f_j \geq \epsilon_1}w'_{v,j} u_{v,j} =
 \left(T-\sum_{j: f_{j}\leq \epsilon_1} u_{v,j} f_{j}\right)\geq \left( T-\frac{\epsilon_1 k}{1-\epsilon_1}\frac{T}{\lambda}\right).
\end{equation}
 Now consider random variables $Y_{v,j}$ for each of these unmatched items:
\begin{equation}
\label{eqn:def}
Y_{v,j} = \begin{cases}  \frac{w'_{v,j} u_{v,j}}{T/\lambda} &: \text{if item $j$ is not matched} \\
0 &: \text{ otherwise} \end{cases}
\end{equation}

Since $u_{v,j} \leq T/\lambda$ and $w_{v,j} \leq f_{v}$, the
$Y_{v,j}$ are random variables bounded in $[0,1]$. Person $v$ is unmatched by the matching
with probability $1-m_v=f_v$. Each such person $v$  gets
a fractional utility of $w'_{v,j} u_{v,j}$ from the small (with respect to the person) item $j$ in the flow graph, if item $j$ is
not matched by the matching. The latter happens with probability $f_{j}$.

Define $G_{v}=\sum_{j}Y_{v,j}$. Then $\frac{T}{\lambda} G_v$ is the total fractional utility after step (b). It
follows from (\ref{eq:1}) that
\begin{eqnarray*}
\expect{G_{v}}&=&\sum_{j}\frac{w'_{v,j}u_{v,j}f_j}{T/\lambda}\geq \epsilon_1 \lambda \left(1-\frac{\epsilon_1 k}{(1-\epsilon_1)\lambda}\right)
\end{eqnarray*}

Thus, since $\lambda = 26 \sqrt{k \ln k}$, we have for sufficiently large $k$ that
\[ \expect{G_{v}} \geq \epsilon_1\lambda \left(1-\frac{\epsilon_1 k}{(1-\epsilon_1)\lambda}\right) \geq 24\ln{k}. \]

That the $Y_{v,j}$'s are negatively correlated follows from Theorem \ref{theorem:neg1}. Therefore, applying Theorem~\ref{thm:chbound-negcorrel}(i)
with $\delta = 1/2$,
\[ \prob{G_v \leq \frac{1}{2} \expect{G_v}} \leq e^{-24\ln{k}/12} = \frac{1}{k^2}; \]
i.e.,
 \begin{equation*}
\prob{\frac{T}{\lambda}G_v \leq \frac{1}{2} \frac{T}{\lambda}\expect{G_v}} \leq \frac{1}{k^2}
\end{equation*}

Hence,
 \begin{equation*}
\prob{\exists v: ~\frac{T}{\lambda}G_v \leq \frac{1}{2} \frac{T}{\lambda}\expect{G_v}} \leq \frac{1}{k}.
\end{equation*}

Therefore the net fractional utility that remains for each person  in the flow graph after scaling is at least $ \frac{1}{2} \frac{T}{\lambda}\expect{G_v}= \frac{1}{2} \frac{T}{26\sqrt{k\ln{k}}}12\ln{k} \geq \frac{T}{5}\sqrt{\frac{\ln{k}}{k}}$, with probability at least $1-\frac{1}{k}$.
\end{proof}

\begin{lemma}
\label{lemma:item}
Fix any item $j$ that is unmatched after Step 2. After the matching and the scaling (step (b)), $j$ has a total fractional incident edge-weight
from the unmatched persons to be at most $\psi(k)$,  with probability at least $1-\frac{1}{k^3}$.
\end{lemma}

\begin{proof}
Note that for any person $v$ for which $j$ is small for $v$, $w_{v,j} \leq f_{v}$; hence,
$w'_{v,j}=\frac{w_{v,j}}{f_{v}}\leq 1$. Define a random variable $Z_{v,j}$ for each person $v$ as:
\begin{equation}
\label{eqn:itemdef}
Z_{v,j} = \begin{cases} w'_{v,j} &: \text{if person $v$ is not matched} \\
0 &: \text{ otherwise} \end{cases}
\end{equation}

Let $X_{j}=\sum_{v} Z_{v,j}$. Then $X_{j}$ is the total weight of all the edges incident on item $j$ in the flow graph
after scaling and removal of all matched persons. We have $\expect{X_{j}}= \sum_{v}w'_{v,j} f_{v}=\sum_{v} w_{v,j} \leq 1$.
The fact that the variables $Z_{v,j}$ are negatively correlated follows from Corollary \ref{cor:neg2}. Thus, applying
Theorem~\ref{thm:chbound-negcorrel}(ii) with $\mu = 1$ and $\delta = \psi(k)$ along with (\ref{eqn:Phi-bound}), we obtain
\[
\prob{X_{j} \geq \psi(k)} \leq \frac{1}{k^3}.
\]
This completes the proof.
\end{proof}

Recall the third step, step (c), of {\bf Allocating Small Items}. Any job in the remaining flow graph with total weight of
incident edges more than $\psi(k)$ is discarded in this step. We now calculate the utility that remains for
each person in the flow graph after step (c).

\begin{lemma}
\label{lemma:stepc}
After removing all the items that have total degree more than $\psi(k)$ in the flow graph, that is after step (c)
of {\bf Allocating Small Items}, the probability that all unmatched persons have remaining utility in the flow graph at least $\sqrt{\frac{\ln{k}}{k}}\frac{T}{2*(3+o(1))}$ is at least $1-\frac{2}{k}$.
\end{lemma}
\begin{proof}
Fix a person $v$ and consider the utility that $v$ obtains from the fractional assignments in the flow graph before step (c). It is at least $\sqrt{\frac{\ln{k}}{k}}\frac{T}{5}$ from Lemma \ref{lemma:small}.
Define a random variable for each item that $v$ claims with nonzero value in the flow graph at step (b):

\begin{equation}
\label{eqn:newdef}
Z'_{v,j} = \begin{cases} u_{v,j} &: \text{if item $j$ has total weighted degree at least $\psi(k)$ } \\
0 &: \text{ otherwise} \end{cases}
\end{equation}

We have $\prob{Z'_{v,j}=u_{v,j}} \leq \frac{1}{k^3}$ from Lemma \ref{lemma:item}.
Therefore, the expected utility for $v$ from all the items in the flow graph that have total incident weight
more than  $\psi(k)$  is at most $\frac{T}{k^3}$. By Markov's
inequality, the probability that the utility for $v$ from the discarded items is more than $\frac{T}{k}$, is at most $\frac{1}{k^2}$. Applying
the union bound, the probability of the utility from the discarded items being more than $\frac{T}{k}$ for some person,
is at most $\frac{1}{k}$. The initial utility before step (c) was at least $\sqrt{\frac{\ln{k}}{k}}\frac{T}{5}$ with probability $1-\frac{1}{k}$.
Thus after step (c), the remaining utility is at least $\sqrt{\frac{\ln{k}}{k}}\frac{T}{5}-\frac{T}{k}$ with probability at least $1-\frac{2}{k}$.
\end{proof}

The next and the final step (d) of allocations is to run \cite{dani:05} on a {\em scaled-down} flow graph. 
The weight on the remaining edges is scaled down by a factor of $\psi(k)$ and hence
for every item node that has not been matched after step (c), the total edge-weight incident on it is at most $one$. Hence after scaling down the utility of any person $v$ in the flow graph is
$\sum_{j} u_{v,j} W_{v,j} \geq \frac{\ln{\ln{k}}}{\ln{k}}\sqrt{\frac{\ln{k}}{k}}\frac{T}{18(1+o(1))}=
\frac{\ln{\ln{k}}}{\sqrt{k\ln{k}}}\frac{T}{18(1+o(1))}$, where $W_{v,j}$ denote the scaled down weight on the edge $(v,j)$. Also, note that the maximum utility of any item in the flow graph is at most $\frac{T}{\lambda}=\frac{T}{26\sqrt{k\ln{k}}}$. Hence, by running the algorithm of \cite{dani:05}, which is a simpler version of Theorem~\ref{thm:matching}, we get the following lemma.

\begin{lemma}
\label{lemma:unsat}
For all persons unmatched by the matching, the total utility received is at least $\Omega(\frac{\ln{\ln{k}}}{\sqrt{k\ln{k}}}T)$ after step (d) with probability at least $1-\frac{2}{k}$.
\end{lemma}
 \begin{proof}
 Let $W_{v,j}$ denote the fractional weight on the scaled down flow graph. Then for every item $j$ in the flow graph, $\sum_{v} W_{v,j} \leq 1$. And for every person $v$ considering the items in the flow graph, $\sum_{j} W_{v,j} u_{v,j} \geq \frac{\ln{\ln{k}}}{\sqrt{k\ln{k}}}\frac{T}{18(1+o(1))}$ with probability at least $1-\frac{2}{k}$. We can now employ the rounding algorithm of \cite{dani:05} which is a simplification of Theorem \ref{thm:matching} without any capacity constraint. We get an integer solution where each person receives a utility of at least $\frac{\ln{\ln{k}}}{\sqrt{k\ln{k}}}\frac{T}{18(1+o(1))}-\frac{T}{\lambda}$, and every item is assigned to at most one person. Since $\lambda = 26\sqrt{k\ln{k}}$, we get the desired result.
\end{proof}
~\\
{\it Theorem~\ref{thm:santa1}
Given any feasible solution to the configuration LP, it can be rounded to a feasible integer solution such that every person gets  at least $\Theta\left(\frac{1}{\sqrt{k\ln{k}}}\right)$ fraction of the optimal utility with probability at least $1-\Theta(\frac{1}{k})$, in polynomial time.
}
\begin{proof}
Any person that is matched by step $2$ of the algorithm {\bf Allocating Big Items} receives a utility of $\frac{T}{\lambda}$. From Lemma~\ref{lemma:unsat}, each person unmatched by the matching receives a utility of $\Omega(\frac{\ln{\ln{k}}}{\sqrt{k\ln{k}}}T)$ with probability at least $1-\frac{2}{k}$. Noting that $\lambda=26\sqrt{k\ln{k}}$, we therefore, get the theorem.
\end{proof}

Thus, our approximation ratio is $\Theta(\frac{1}{\sqrt{k\ln{k}}})$.
This provides an upper bound of $O(\sqrt{k\ln{k}})$ on the integrality gap of the configuration LP for max-min fair allocation, nearly
matching the  lower bound of $\Omega(\sqrt{k})$ due to \cite{bansal:stoc06}. 

\section{Designing Overlay Multicast Networks For Streaming}
\label{section:spaa}
The work of \cite{DBLP:conf/spaa/AndreevMMS03} studies approximation algorithms
for designing a multicast overlay network. We first describe the problem and state the results in \cite{DBLP:conf/spaa/AndreevMMS03} (Lemma \ref{lem:spaa1} and Lemma \ref{lem:spaa2}). Next, we show our main improvement in Lemma \ref{lem:color}.

\subsection{Background} The background text here is largely borrowed from \cite{DBLP:conf/spaa/AndreevMMS03}. An overlay network can be represented as a tripartite digraph $N=(V,E)$. The nodes $V$ are partitioned into sets of entry points called sources ($S$), reflectors ($R$), and edge-servers or sinks ($D$). There are multiple commodities or streams, that must be routed from sources, via reflectors, to the sinks that are designated to serve that stream to end-users. Without loss of generality, we can assume that each source holds a single stream. 
 There is a cost associated with usage of every link and reflector. 
 There are capacity constraints, especially on the reflectors,
that dictate the maximum total bandwidth (in bits/sec) that the reflector is allowed to send. 
To ensure reliability, multiple copies of each stream may be sent to the designated edge-servers.
\begin{table*}
\begin{eqnarray}
\text{min} && \sum_{i \in R} r_{i}z_{i} +\sum_{i \in R}\sum_{k \in S}c_{k,i,k}y_{i,k}+\sum_{i \in R}\sum_{k \in S}\sum_{j \in D}c_{i,j,k}x_{i,j,k} \nonumber\\
s.t. &&\\
&& y_{i,k} \leq z_{i}~~\forall i \in R,~~\forall k \in S \label{eqn:cons1} \\
&& x_{i,j,k} \leq y_{i,k}~~\forall i \in R, ~~\forall j \in D, ~~\forall k \in S \label{eqn:cons2} \\
&& \sum_{k \in S} \sum_{j \in D} x_{i,j,k} \leq F_{i}z_{i}~~\forall i \in R
\label{eqn:fanout} \\
&& \sum_{i \in R}x_{i,j,k}w_{i,j,k} \geq W_{j,k}~~\forall j \in D, \forall k \in S \label{eqn:weight} \\
&& x_{i,j,k} \in \{0,1\}, y_{i,k} \in \{0,1\}, z_{i} \in \{0,1\}
\label{eqn:integrality}
\end{eqnarray}
\caption{Integer Program for Overlay Multicast Network Design}
\label{table:spaa}
\end{table*}

All these requirements can be captured by an integer program. Let us use indicator variable $z_{i}$ for building reflector $i$, $y_{i,k}$ for delivery of $k$-th stream to the $i$-th reflector and $x_{i,j,k}$  for delivering $k$-th stream to the $j$-th sink through the $i$-th reflector. $F_{i}$ denotes the fanout constraint for each reflector $i \in R$. Let $p_{x,y}$ denote the failure probability on any edge (source-reflector or reflector-sink). We transform the probabilities into weights: $w_{i,j,k}=-\log{(p_{k,i}+p_{i,j}-p_{k,i}p_{i,j})}$. Therefore, $w_{i,j,k}$ is the negative log of the probability of a commodity $k$  failing to reach sink $j$ via reflector $i$. On the other hand, if $\phi_{j,k}$ is the minimum required success probability for commodity $k$ to reach sink $j$, we instead use $W_{j,k}=-\log{(1-\phi_{j,k})}$. Thus $W_{j,k}$ denotes the negative log of maximum allowed failure. $r_i$ is the cost for opening the reflector $i$ and $c_{x,y,k}$ is the cost for using the link $(x,y)$ to send commodity $k$. Thus we have the IP (see Table \ref{table:spaa}).

Constraints (\ref{eqn:cons1}) and (\ref{eqn:cons2}) are natural consistency
requirements; constraint (\ref{eqn:fanout}) encodes the fanout restriction.
Constraint (\ref{eqn:weight}), the \emph{weight} constraint, ensures
quality and reliability. Constraint (\ref{eqn:integrality}) is the standard
integrality-constraint that will be relaxed to construct the LP relaxation.

There is an important  stability requirement that is referred as \emph{color constraint} in \cite{DBLP:conf/spaa/AndreevMMS03}. Reflectors are grouped into $m$ color classes, $R=R_{1} \cup R_{2} \cup \ldots \cup R_{m}$. We want each group of reflectors to deliver not more than one copy of a stream into a sink. This constraint translates to
\begin{equation}
\label{eqn:color}
\sum_{i \in R_{l}} x_{i,j,k} \leq 1~\forall j \in D, ~\forall k \in S, ~\forall l \in [m]
\end{equation}

Each group of reflectors can be thought to belong to the same ISP. Thus we want to make sure that a client is served only with one -- the best -- stream possible from a certain ISP. This diversifies the stream distribution over different ISPs and provides stability. If an ISP goes down, still most of the sinks will be served. We refer the LP-relaxation of  integer program
(Table \ref{table:spaa}) with the color constraint (\ref{eqn:color})
as \textbf{LP-Color}.

All of the above is from \cite{DBLP:conf/spaa/AndreevMMS03}.
The work of \cite{DBLP:conf/spaa/AndreevMMS03} uses a two-step rounding procedure and obtains the following guarantee.

First stage rounding: Rounds $z_{i}$ and $y_{i,k}$ for all $i$ and $k$ to decide which reflector should be open and which streams should be sent to a reflector. The results here can be summarized in the following
lemma:

\begin{lemma}
\textbf{(\cite{DBLP:conf/spaa/AndreevMMS03})}
\label{lem:spaa1}
The first-stage rounding algorithm incurs a cost at most a factor of $64\log{|D|}$ higher than the optimum cost, and with high probability violates the weight constraints by at most a factor of $\frac{1}{4}$ and the fanout constraints by at most a factor of $2$. Color constraints are all satisfied.
\end{lemma}


Second stage rounding: Rounds $x_{i,j,k}$'s using the open reflectors and streams that are sent to different reflectors in the first stage. The results in this stage can be summarized as follows:

\begin{lemma}
\textbf{(\cite{DBLP:conf/spaa/AndreevMMS03})}
\label{lem:spaa2}
The second-stage rounding incurs a cost at most a factor of $14$ higher than the optimum cost and violates each of fanout, color and weight constraint by at most a factor of $7$.
\end{lemma}

\subsection{Main Contribution}

Our main contribution is an improvement of the second-stage rounding
through the use of repeated \textbf{RandMove} and by judicious choices of
constraints to drop. Let us call the linear program that remains just at the end of first stage \textbf{LP-Color2}:

\begin{eqnarray*}
&&\min{\sum_{i \in R}\sum_{k \in S}\sum_{j \in D} c_{i,j,k}x_{i,j,k}}\\
&& \text{s.t.}\\
&& \sum_{k \in S}\sum_{j \in D}x_{i,j,k} \leq F_{i}~\forall i \in R ~(\text{Fanout})\\
&& \sum_{i \in R}x_{i,j,k}w_{i,j,k} \geq W_{j,k}~\forall j \in D, \forall k \in S ~(\text{Weight})\\
&& \sum_{i \in R_{l}} x_{i,j,k} \leq 1~\forall j \in D, ~\forall k \in S, ~\forall l \in [m]~(\text{Color})\\
&& x_{i,j,k} \in \{0,1\} ~\forall i \in R,\forall j \in D, \forall k \in S\\
\end{eqnarray*}

We show:

\begin{lemma}
\label{lem:color}
\textbf{LP-Color2} can be efficiently rounded such that the cost and weight constraints are satisfied exactly, fanout constraints are violated at most by additive $1$, and the color constraints are violated at most by additive $3$.
\end{lemma}

The proof is very similar to Theorem~\ref{thm:matching}. Note that, here instead of having capacity constraints, we have fanout constraints. Weight constraints correspond to load constraints in Theorem~\ref{thm:matching}, but now they provide lower bounds. Moreover, the color constraints can be thought of as additional capacity constraints imposed on a set of reflectors. This constitutes the main change from Theorem~\ref{thm:matching}, and we need new conditions to drop color constraints $(D2'')$. The color constraints being all disjoint help us in the rounding.

\begin{proof}
Let $x^{*}_{i,j,k} \in [0,1]$ denote the fraction of stream generated from source $k \in S$ reaching destination
$j \in D$ routed through reflector $i \in R$ after the first stage of rounding. Initialize $X=x^*$. The algorithm
consists of several iterations. the random value at the end of iteration $h$ is denoted by $X^{h}$. Each iteration $h$ conducts a randomized update using {\bf RandMove} on the polytope of a linear system constructed from a subset of constraints of {\bf LP-Color2}. Therefore by induction on $h$, we will have for all $(i,j,h)$ that $\expect{X_{i,j}^{h}}=x^{*}_{i,j}$. Thus the cost constraint is maintained exactly on expectation. The entire procedure can be derandomized by the method of conditional probabilities, yielding the required bounds on the cost.

Let $R$ and $SD$ denote the set of reflectors and (source, destination) pairs respectively.
Suppose we are at the beginning of some iteration $(h+1)$ of the overall algorithm and currently looking
at the values $X_{i,j,k}^{h}$. We will maintain two invariants:
\begin{description}
\item[(I1'')] Once a variable $x_{i,j,k}$ gets assigned to $0$ or $1$, it is never changed;
\item[(I2'')] Once a constraint is dropped in some iteration, it is never reinstated.
\end{description}
Iteration $(h+1)$ of rounding consists of three main steps:

\begin{enumerate}
\item Since we aim to maintain ({\bf I1''}), let us remove all $X_{i,j,k}^{h} \in \{0,1\}$; i.e.,
we project $X^{h}$ to those coordinates $(i,j,k)$ for which $X_{i,j,k}^{h} \in (0,1)$, to obtain the
current vector $Y$ of floating (yet to be rounded) variables; let $\mathcal{S}\equiv(A_{h}Y=u_{h})$ denote
the current linear system that represents {\bf LP-Color2}. In particular, the fanout constraint
for a reflector in $\mathcal{S}$ is its residual fanout $F'_{i}$; i.e., $F_{i}$ minus the
number of streams that are routed through it.
\item Let $v$ denote the number of floating variables, i.e., $Y\in (0,1)^{v}$. We now drop the following constraint:
    \begin{description}
    \item[(D1'')] Drop fanout constraint for degree $1$ reflector denoted $R_1$, i.e, reflectors with only one floating variable associated with it. For any degree $2$ reflectors denoted $R_2$, if
        it has a tight fanout of $1$ drop its fanout constraint.
    \item[(D2'')] Drop color constraint for a group of reflectors $R_{l}$, if they have at most four floating variables associated with them.
    \end{description}
    \end{enumerate}
Let $\mathcal{P}$ denote the polytope defined by this reduced system of constraints.
A key claim is that $Y$ is not a vertex of $\mathcal{P}$ and thus we can apply
{\bf RandMove} and make progress either by rounding a new variable or by dropping
a new constraint. We count the number of variables $v$ and the number of tight constraints
 $t$ separately. We have $$t=\sum_{i \in R \setminus R_1} 1+\sum_{k \in S}\sum_{j \in D}(l_{k,j}+1),$$ where
 $l_{j,k}$ is the number of tight color constraints for the stream generated at source $k$
 and to be delivered to the  destination $j$.  We further have $v \geq \sum_{i \in R} (F_i +1)$, and that
$v \geq \sum_{k \in S, \j \in D, l_{k,j} > 0} 4l_{k,j} + \sum_{k \in S, \j \in D, l_{k,j}= 0} 2$. Thus by averaging, $$v \geq \frac{\sum_{i \in R} (F_i +1)}{2}+  \sum_{k \in S, \j \in D, l_{k,j} > 0} 2l_{k,j}+\sum_{k \in S, \j \in D, l_{k,j}= 0} 1.$$ A moment's reflection shows that the
 system can become underdetermined only if there is no color constraint associated with a stream $(j,k)$, each reflector $i$ has two floating variables associated with it with total contribution $1$ towards fanout and each stream $(j,k)$ is routed fractionally through two reflectors. But in this situation all the fanout constraints are dropped violating fanout at most by an additive one and making the system underdetermined once again.
The color constraints are dropped only when there are less than four
floating variables associated with that group  of reflectors; hence, the color constraints can get violated at most by an additive $3$. The fanout constraint is dropped only for singleton reflectors or degree-2 reflectors with fanout equaling $1$. Hence the fanout is violated only by an additive
excess of $1$. The weight constraint is never dropped, and
is hence maintained exactly.
\end{proof}

\section{Appendix}
Here we give a proof of Theorem~\ref{theorem:neg1}. This establishes the negative correlation property of bipartite dependent rounding on random variables defined on multiple vertices and used in Section~\ref{sec:max-min}.

{\it Theorem \ref{theorem:neg1}.
Define an indicator random variable $z_{j}$ for each item $j \in B$ with $m_{j} < 1$, such that $z_{j}=1$
 if item $j$ is matched by the matching. Then, the indicator random variables $\{z_{j}\}$ are negatively correlated.
}
 \begin{proof}
 Consider any collection of items $j_1,j_2,\ldots,j_{t}$. Let $b=1$ (the proof for the case $b=0$ is identical).
 Let $y_{i,j,k}$ denote the value of $y_{i,j}$ at the beginning of the $k$-th iteration of bipartite dependent rounding.
 Define, $z_{j,k}=\sum_{i, (i,j) \in M} y_{i,j,k}$. Clearly, $z_{j}=\sum_{i, (i,j) \in M} y_{i,j,|M|+1}$. We will show that
 \begin{equation}
 \label{equation:relation}
 \forall k, \expect{\prod_{i=1}^{t} z_{j_i,k}} \leq \expect{\prod_{i=1}^{t} z_{j_i,k-1}}
 \end{equation}

 Thus, we will have
 \begin{eqnarray*}
 \prob{\bigwedge_{i=1}^{t} z_{j_i}=1} &=&\expect{\prod_{i=1}^{t} z_{j_i,|M|+1}} \\
 && \leq \expect{\prod_{i=1}^{t} z_{j_i,1}}\\
 &=&\prod_{i=1}^{t} \sum_{v} y_{v,j_{i},1}= \prod_{i=1}^{t} m_{j_{i}}= \prod_{i=1}^{t} \prob{z_{j_{i}}=1}
 \end{eqnarray*}

 We now prove (\ref{equation:relation}) for a fixed $k$. Note that any vertex that is not the end point of
 the maximal path or the cycle on which dependent rounding is applied in the $k-1$-st round, retains its previous
 $z$ value. There are three cases to consider:

{\bf Case 1}: {\it Two vertices among $j_{1},j_{2},\ldots,j_{t}$ have their values modified}. Let these vertices be
say $j_1$ and $j_2$. Therefore, these two vertices must be the end points of the maximal path on which dependent rounding
is applied on the $k-1$-st round. The path length must be even. Let $B(j_1,j_2, \alpha,\beta)$ denote the event that the
jobs $\{j_1,j_2\}$ have their values modified in the following probabilistic way:
\begin{eqnarray*}
(z_{j_1,k}, z_{j_2,k}) = \begin{cases} (z_{j_1,k-1}+\alpha, z_{j_2, k-1}-\alpha) & \text{ with probability } \frac{\beta}{\alpha+\beta}\\
  (z_{j_1,k-1}- \beta,z_{j_2,k-1}+\beta) & \text{ with probability } \frac{\alpha}{\alpha+\beta} \end{cases}
\end{eqnarray*}

Thus
\begin{eqnarray*}
&& \expect{\prod_{i=1}^{t} z_{j_i,k}|\forall i \in [1,t],~ z_{j_{i},k-1}=a_{j_{i}} \wedge B(j_1,j_2,\alpha,\beta)}\\
&=& \expect{z_{j_1,k}z_{j_2,k}|\forall i \in [1,t],~ z_{j_i,k-1}=a_{j_{i}} \wedge B(j_1,j_2,\alpha,\beta)}\prod_{i=3}^{t}a_{j_{i}}
\end{eqnarray*}

The above expectation can be written as $(\eta+\gamma)\Pi_{i=3}^{t}a_{j_{i}}$, where
$$\eta= (\beta/(\alpha+\beta))(a_{j_1}+\alpha)(a_{j_2}-\alpha),\, \text{and}$$
$$ \gamma=(\alpha/(\alpha+\beta))(a_{j_1}-\beta)(a_{j_2}+\beta).$$

Now, it can be easily seen that $\eta+\gamma \leq a_{j_1}a_{j_2}$. Thus for any fixed $j_1,j_2$ and for any fixed
$(\alpha,\beta)$, and for fixed values $a_{f}$ the following holds:

$$\expect{\prod_{i=1}^{t}z_{j_i,k}|\forall i \in [1,t],~ z_{j_i,k-1}=a_{j} \wedge B(j_1,j_2,\alpha,\beta)} \leq \prod_{i=1}^{t}a_{j}.$$

Hence, $\expect{\prod_{i=1}^{t} z_{j_i,k}} \leq \expect{\prod_{i=1}^{t} z_{j_i,k-1}}$ here.

{\bf Case 2}: {\it One vertex among $j_{1},j_{2},\ldots,j_{t}$ has its value modified}. Let the vertex be
$j_1$ say. Therefore, this vertex must be the end point of the maximal path on which dependent rounding
is applied on the $(k-1)$-st round. The path length must be odd. Let $B(j_1,\alpha,\beta)$ denote the event that the job $j_1$ has its value modified in the following probabilistic way:
\begin{eqnarray*}
z_{j_1,k} = \begin{cases} z_{j_1,k-1}+\alpha & \text{ with probability } \frac{\beta}{\alpha+\beta}\\
  z_{j_1,k-1}- \beta & \text{ with probability } \frac{\alpha}{\alpha+\beta} \end{cases}
\end{eqnarray*}

Thus,
$$\expect{z_{j_1,k}|\forall i \in [1,t],~ z_{j_i,k-1}=a_{j} \wedge B(j_1,\alpha,\beta)}=a_{j_1}.$$ Since the values of $z_{j_i}$, $i \in [2,t]$ remains unchanged and the above equation holds for any $j_1,\alpha, \beta$, we have $\expect{\prod_{i=1}^{t} z_{j_i,k}} \leq
\expect{\prod_{i=1}^{t} z_{j_i,k-1}}$.

\textbf{Case 3}: {\it None among $j_{1},j_{2},\ldots,j_{t}$ has its value modified}.

In this case, the value of $z_{j_i,k}$'s, $i \in [1,t]$, do not change. Hence, $\expect{\prod_{i=1}^{t} z_{j_i,k}} \leq \expect{\prod_{i=1}^{t} z_{j_i,k-1)}}$.
 \end{proof}

\medskip \noindent \textbf{Acknowledgment.} We thank Chandra Chekuri,
R.\ Ravi, Mohit Singh, and Jan Vondr\'{a}k for helpful discussions on 
dependent and iterative rounding. We also thank the anonymous reviewers for detailed comments that significantly helped us to improve the presentation.

\bibliographystyle{abbrv}
\bibliography{af2009,af2}



\end{document}